\RequirePackage{fix-cm}
\documentclass[twocolumn]{svjour3}          
\smartqed  
\usepackage{graphicx}

\usepackage[linesnumbered,ruled,vlined]{algorithm2e} 
\usepackage{amsmath}   
\usepackage{amssymb}
\usepackage{enumitem}  
\usepackage{multirow}
\usepackage{epstopdf}
\begin{document}

\title{($\alpha,k$)-Minimal Sorting and Skew Join in MPI and MapReduce
}
\subtitle{
}


\author{Silu Huang         \and
        Ada Wai-Chee Fu 
}


\institute{Silu Huang \at
              Chinese University of Hong Kong \\
              \email{slhuang@cse.cuhk.edu.hk}           
           \and
           Ada Wai-Chee Fu \at
              Chinese University of Hong Kong \\
              \email{adafu@cse.cuhk.edu.hk}
}

\date{Received: date / Accepted: date}

\maketitle

\begin{abstract}

As computer clusters are found to be highly effective for handling massive datasets, the design of efficient parallel algorithms for such a computing model is of great interest.
 We consider $(\alpha,k)$-minimal algorithms for such a purpose, where $\alpha$ is the number of rounds in the algorithm, and $k$ is a bound on the deviation from perfect workload balance. We focus on new $(\alpha,k)$-minimal algorithms for sorting and skew equi-join operations for computer clusters. To the best of our knowledge the proposed sorting and skew join algorithms achieve the best workload balancing guarantee when compared to previous works. Our empirical study shows that they are close to optimal in workload balancing. In particular, our proposed sorting algorithm is around 25\% more efficient than the state-of-the-art Terasort algorithm and achieves significantly more even workload distribution by over 50\%. 
\end{abstract}

\section{Introduction}

A \textbf{Computer cluster} consists of a set of connected computers or nodes
usually connected to each other through a local area network (LAN).
Cluster computing has emerged as a commonly used infra-structure for efficient
big data computation because of the elasticity of the cluster size
using high speed networks and low cost CPUs.
%
%
In a computer cluster, the machines are
isolated and each has its own memory and storage.

\textbf{MPI} (Message Passing Inteface) is a widely used standard for communication among nodes in a cluster \cite{Gropp94}. 
MPI provides a clearly defined base set of routines that can be used to build high level parallel algorithms running on computer clusters. MPI programs work with processes.
Typically, for maximum performance, each CPU (or core in a multi-core machine) will be assigned just a single process. We shall refer to the processes as machines in this paper.

\textbf{MapReduce }\cite{Dean04OSDI} is 
programming model for parallel computing.
%
%
It has been found useful for processing large datasets in a parallel and distributed architecture, typically on a computer cluster.
 MapReduce hides details of mechanisms in data distribution, fault tolerance and certain amount of
load balancing, so that the implementation for problem solving can be greatly simplified.
Each MapReduce job consists of the following phases.

\begin{itemize}[leftmargin=0.1in]
\item[-]
Map -- at each mapper a map function is applied to each input record $(k1,v1)$ to generate key, value pairs of the form $(k2,v)$.
\item[-]
Shuffling (Mapper to Reducer) -- output of mappers are distributed to reducers. $(k2,v)$ of the same key $k2$ are sent to the same reducer.
\item[-]
Reduce -- works on all $(k2,v_1), (k2, v_2), ... (k2,v_j)$ on all key value pairs for the same key $k2$. Output are sent to a distributed file system (DFS).
\end{itemize}

Apache Hadoop MapReduce is a programming model for computer clusters. Hadoop utilizes a
\textbf{distributed file system (DFS)}. DFS
is often used in a computer cluster setup, and typically supports replication and
fault tolerance. Both MPI and MapReduce can be built on
a DFS such as the
 Hadoop DFS (HDFS).

 \subsection{The Sorting and Skew Join Problems}

 We consider the design of parallel algorithms for the basic data management problems of sorting and join operation on two tables with skew key distributions, with the models of MPI and MapReduce for computer clusters. Sorting is a fundamental problem useful in many applications.
 While the state-of-the-art sorting algorithm, Terasort \cite{Malley08Yahoo}, can outperform Hadoop's default sorting, the workload distribution is not even. From experiments in \cite{Tao13sigmod}, the maximum workload of a machine is around 1.6 times that of the optimal distribution.
 We propose a new parallel algorithm SMMS that attains the best theoretical guarantee on workload distribution to our knowledge. SMMS adopts a strategy different from Terasort in that the data is evenly distributed to all machines and each machine first sorts its assigned data portion.
 Our algorithm is deterministic in that no random sampling takes place. From experiments, we show that the workload distribution of SMMS is very close to optimal in all test cases. As a result, SMMS is consistently more efficient than Terasort.

We have chosen the second problem of skew join for our study because it has been a challenging
problem for load balancing.
Data skew has been shown to cause sub-linear speedup \cite{DeWitt92CACM}.
In the recent development of the Apache Pig system on top of MapReduce,
it has been found that data skew in join is a serious
problem: ``\emph{
... we have experienced
performance problems due to data skew with Pig at Yahoo. One particularly
challenging scenario occurs when a join is performed, and a few of the join
keys have a very large number of matching tuples, which must be handled
by a single node in our current implementation,
...
}''\cite{Gates09vldb}.
Such skewness in the output of join is classified as Join Product Skew (JPS) in \cite{Walton91vldb}.
A study of a number of algorithms for handling join skew is conducted in \cite{DeWitt92vldb} and Apache
Pig \cite{Gates13Bieee} has adopted the skew hash join method from \cite{DeWitt92vldb}
with histogram based estimation.
As discussed in \cite{Gates13Bieee}: ``\emph{achieving an even balance of work between reducers is not always feasible with skewed join. If the distribution of join keys in the right side input is skewed, the work load of reducers will still be skewed}''.
In this paper, we propose a randomized algorithm and a deterministic algorithm for handling join skew. Both algorithms achieve the best theoretical bound for even workload distribution to our knowledge. From experiments we show that the workload distribution is close to optimal.

\subsection{The Notion of $(\alpha,k)$-Minimal Algorithms}

A lot of interests have focused on the development of efficient parallel algorithms in
a computer cluster environment.
There are multiple factors that affect the performance, which include
the CPU computation costs, I/O costs and network transmission costs, and these costs in turn
depend on the factors of load balancing and job sequencing.
We aim to derive from these factors a model with key properties for an effective algorithm.
With a computer cluster of $t$ machines, the ideal goal is to have a $t$ fold increase in the performance compared to a single machine.
However, it is not an obvious task to convert a sequential algorithm to a parallel algorithm with optimal speedup. Often, when the workload is not balanced among the $t$ machines, the last machine to finish its assigned tasks will delay the completion of the entire execution. This phenomenon has been dubbed ``the curse of the last reducer'' in \cite{Suri11www} when MapReduce is adopted in the parallel computation.
MPI and MapReduce are two common standards for cluster computing.

Since a parallel algorithm may not be perfect,
we introduce the notion of an
$(\alpha,k)$-minimal algorithm that quantifies the properties in a parallel computation for the analysis of such an algorithm. $k$ is a positive real number that indicates the deviation from perfect load balancing, while $\alpha$ is the number of rounds under either the MPI or MapReduce model.
We subject our proposed algorithms to this yardstick and show that our algorithms are $(\alpha,k)$-minimal with values of
$\alpha$ of 3 or below, and values of $k$ of around 2.

\subsection{Main Contributions}

Our main contributions are summarized as follows: (1) We introduce the notion of an $(\alpha,k)$-minimal algorithm for computer clusters, with consideration of the MPI and MapReduce models.
(2) We propose a new sorting algorithm called SMMS.
To the best of our knowledge SMMS has the best load balancing guarantee among all previous works. The theoretical load balancing bound is confirmed by our empirical studies in which we show that the workload distribution is close to optimal in all test cases.
(3) We propose two algorithms, RandJoin and StatJoin, for equi-join to handle skew keys in the given relations. To our knowledge, RandJoin achieves the best theoretical guarantee on even load distribution with high probabilities, and StatJoin achieves the best such guarantee deterministically.
(4)
We show that
SMMS is $(3,k)$-minimal for $k < 2$ with proper settings.
We show that StatJoin is $(3,k)$-minimal for $k \approx 2$
for skew data.
We also show that RandJoin is nearly (1,2)-minimal and Terasort \cite{Malley08Yahoo} is (3,6)-minimal with high probabilities.
(5) We have conducted an extensive set of experiments on a
computer cluster to evaluate the proposed algorithms. SMMS is shown to have nearly perfect workload distribution, and the extra cost for load balancing is negligible. As a result we show overall better performance compared to the state-of-the-art Terasort.
RandJoin and StatJoin are shown to achieve almost perfect workload distribution under different skew key conditions.



This paper is organized as follows: Section \ref{sec:def}
motivates and defines $(\alpha,k)$-minimal algorithms.
Section \ref{sec:sorting} introduces our sorting algorithm called SMMS and also analyzes Terasort. Section \ref{sec:join} is about the join algorithms. Section \ref{sec:exp} reports on our experimental results. Section \ref{sec:related} summarizes related works, and Section \ref{sec:conclusion} concludes the paper.
%
\section{$(\alpha,k)$-Minimal Algorithm}  \label{sec:def}

With a parallel algorithm executed on a cluster with $t$ machines, a major goal is to achieve optimal speedup over the sequential algorithm. Ideally, the speedup would be $t$.
However, this is difficult to achieve due to communication overhead as well as synchronization overhead when some parts of the job must be done serially, one task after another.
Here we introduce the notion of an $(\alpha,k)$-minimal algorithm under the MPI or MapReduce model, where $\alpha$ is the number of synchronized rounds, and $k$ indicates the deviation from perfect load balancing. Let $N$ be the problem size, which is the sum of the input size $N_{in}$ and the output size $N_{out}$.
The sequential workload $W_{seq}$ is given by
$max ( N_{in}, N_{out})$.
Let $t$ be the number of machines (denoted by $\mathcal{M}_1, ..., \mathcal{M}_t$). We expand the definition of an $(\alpha,k)$-minimal algorithm in the following subsections.

\subsection{Number of Rounds, $\alpha$} \label{subsec:iteration}
Generally, problems can not be fully parallelized. In other words, a parallel algorithm often consists of some parts of work that must be done serially, one after another.
Synchronization is needed when such serial order is to be in place. We assume that the given parallel algorithm is executed in rounds for all machines, and synchronization among the machines takes place between 2 consecutive rounds. Synchronization waits for the last machine to complete in the current round before the start of the next round.

For MPI, a synchronization function can be called to ensure the correctness of task sequence and the synchronization overhead is quite small. We refer to the program between two synchronization points as one round.
For MapReduce, each MapReduce job is considered as one round.


For an $(\alpha,k)$-minimal algorithm, the number of rounds is given by $\alpha$.




\subsection{Workload on Each Machine }\label{subsec:workload}
 A good parallel algorithm distributes the total workload among $t$ machines so that all machines complete their tasks at almost the same time. At each round, the last machine to finish will delay the entire execution.
Workload is denoted by $W_i$ on machine $i$. 

An $(\alpha,k)$-minimal algorithm bounds $W_i$ on each machine to within $k$ times that of the optimal workload $W_{seq}/t$,
where $W_{seq}$ = $\max (N_{in}, N_{out})$.
\begin{eqnarray}
\mbox{ \hspace*{25mm}} W_i \leq k \left( W_{seq}/t \right)
\label{ineq:workload}
\end{eqnarray}
Inequality (\ref{ineq:workload}) gives an insight to the required storage space and I/O cost on each machine, where $k$ quantifies the load balancing of the parallel algorithm.
Note that while for many problems, including the sorting problem, the total workload is determined by the input size, there are also problems, such as the skew join problem, where the workload is dominated by the output size.

\subsection{Network Transmission Cost } \label{subsec:network}
Cluster computing is based on a shared-nothing infrastructure, hence, unlike a sequential algorithm, a parallel algorithm has to take network transmission overhead into account.
For MPI, each machine needs to send (receive) data to (from) other machines in order to communicate with each other and distribute tasks to each machine.
For MapReduce, each reducer needs to extract data from the mapper output. When the reducer and the mapper reside on different machines, network transmission is needed in the shuffling process.
When a DFS is used, network transmission is needed to access or store files maintained by the DFS.


Let $N_i$ be the network transmission cost with respect to machine $\mathcal{M}_i$, where $1 \leq i \leq t$. $N_i$ is defined by the volume of data transmitted to and from the machine. Let $N$ be the problem size ($N = N_{in}+N_{out}$).
With an $(\alpha,k)$-minimal algorithm, $N_i$ is at most $k$ times that of $N/t$ at each round.\footnote{Note that we do not consider transmission cost due to file replication in DFS since the replication factor is user defined and varies accordingly.}
\begin{eqnarray}
\mbox{ \hspace*{25mm}} N_i \leq k \left( N/t \right)
\label{ineq:network}
\end{eqnarray}

Inequality (\ref{ineq:network}) guarantees that the network transmission cost is at most $k$ times that of a fair share of the problem size for each machine.
Note that in a typical setting of a computer cluster, the machines are physically located together
and connected with a gigabit ethernet. Currently, an ethernet switch can support up to 48 machines. When all machines are connected by a single ethernet switch, each with a separate port, all machines can transmit in parallel and the speed is comparable to hard disk transfer rates.
Thus, the overall transmission time will depend on the single machine with the maximum transmission volume.
We remark that in many cases, each machine's network transmission volume is closely related to the workload on this machine. This is the case with our sorting and join algorithms, which we will discuss in detail later.

\subsection{Computational Cost on Each Machine } \label{subsec:cost}
Next we consider the computational cost on each machine.
Let $C_i$ be the computational cost at machine $\mathcal{M}_i$, $1 \leq i \leq t$.
We bound the cost by means of the cost of a comparable sequential algorithm, denoted by $C_{seq}$.
We use big $O$ analysis since the hidden constant of the big $O$ notation is small when the algorithms are comparable.
%
An $(\alpha,k)$-minimal algorithm satisfies Equality (\ref{ineq:cpu}) at each round.
\begin{eqnarray}
\mbox{ \hspace*{25mm}}C_i = O\left( C_{seq}/ t \right)
\label{ineq:cpu}
\end{eqnarray}
Equality (\ref{ineq:cpu}) ensures that the overall computational cost of a parallel algorithm is bounded by that of the comparable sequential algorithm, and the cost is distributed evenly.

In summary, an $(\alpha,k)$-minimal algorithm based on MPI or MapReduce consists of $\alpha$ rounds, and satisfies the Inequalities (\ref{ineq:workload}), (\ref{ineq:network}),
and (\ref{ineq:cpu}). Such an algorithm provides for guarantees about the desirable properties for efficient execution.
In the reminder of this paper, we design $(\alpha,k)$-minimal algorithms for the sorting problem and skew join problem, where $\alpha \leq 3$, and $k \leq 2$.

\section{Sorting}  \label{sec:sorting}

We are given a set $\mathcal{S}$ of $n$ objects, where each object is a real number.
Our goal is to sort the $n$ objects with a computer cluster.
For simplicity the objects themselves are the sort keys.
Let there be $t$ machines in the cluster, namely,
$\mathcal{M}_1, \mathcal{M}_2, ..., \mathcal{M}_t$.
For simplicity we assume that $n$ is a multiple of $t$, and let
$m = n/t$. This assumption can
be easily removed by padding some dummy objects to $S$.
We also assume that initially the $n$ objects are evenly distributed to the $t$ machines, so that each machine is assigned $m$ objects.
Note that by the definition of sets, there is no duplicated key in $S$.
We shall discuss about how to sort a bag of objects in Section
\ref{sec:sortdiscuss}.

\subsection{SMMS sorting}

Our proposed parallel algorithm is called Sort-Map-Merge Sorting (SMMS) and it involves the main steps of sorting, mapping and merging. We have implemented the algorithm on MPI and there are 3 rounds in the algorithm. As mentioned earlier, each machine is assigned $m$ objects initially.

In the first round, each machine samples $s+1$ objects as follows.
The $m = n/t$ objects $\mathcal{S}_i$ in each machine $\mathcal{M}_i$ are sorted and divided into $s$ equi-depth (equi-frequency) intervals.
Let the objects received by $\mathcal{M}_i$ in sorted order be
$o_1, o_2, ..., o_m$.
$\mathcal{M}_i$ picks $s+1$ sample objects
$\lambda_{i,0}, \lambda_{i,1}, ..., \lambda_{i,s}$,
where $\lambda_{i,0} = o_1$,
and
$\lambda_{i,j}$ is the $\lceil j*m/s\rceil$-th smallest object in $S_i$.
Thus,
$\lambda_{i,1}$
=
$o_{\lceil m/s\rceil}$,
$\lambda_{i,2}$=
$o_{\lceil 2m/s \rceil}$, ...,
$\lambda_{i,s}$ =
$o_m$. Thus,
$s+1$ is the sampling size, and $s$ is a multiple of $t$.
Let
\begin{eqnarray}
\mbox{ \hspace*{25mm}}s = rt
\end{eqnarray}
where $r \geq 1$ is a small integer.
The sampled objects are sent to machine $\mathcal{M}_1$.

In Round 2, $\mathcal{M}_1$ collects all the sample objects from every machine and
then computes $t+1$ global key boundaries $b_0, b_1, ..., b_t$,
so that each interval $[b_i, b_{i+1})$ forms a bucket $\beta_{i+1}$ and the intervals partition the data set. Each data objects belongs to one bucket.
The algorithm to compute the boundaries will be described in Section \ref{sec:alg1}.
The boundaries are sent to all machines.
In Round 3, each machine distributes the sorted data $S_i$
according to the bucket boundaries, so that data belonging to bucket
$\beta_i$ go to machine $\mathcal{M}_i$.
$\mathcal{M}_i$ merges the data coming from other machines to form
the sorted list for bucket $\beta_i$. The sorted lists from all machines form the sorted result set.
The pseudo code for SMMS is given in Figure \ref{fig:smmsort}.
Note that if implemented in MapReduce, then the first two rounds form one MapReduce job and Round 3 forms another MapReduce job.

\begin{figure}[htbp]
\medskip
\hrulefill
\\
{\texttt{\bf SMM Sorting - a deterministic algorithm}}

\vspace*{-2mm}
\hrulefill

\smallskip

{\bf Round 1:}
  $\mathcal{S}$ is evenly distributed among $t$ machines.
  Each machine $\mathcal{M}_i$ handle a subset $S_i \subset \mathcal{S}$, where $|S_i|=n/t=m$. 
  On each $\mathcal{M}_i$, sort subset $S_i$ locally and pick $\lambda_{i,0}$, $\lambda_{i,1}$, ...,$\lambda_{i,s}$ and send to machine $\mathcal{M}_1$, where
  $\lambda_{i,0}$ is the smallest object in $S_i$ and for $j > 0$,
  $\lambda_{i,j}$ is the $\lceil j*m/s\rceil$-th smallest object in $S_i$.

\medskip

\textbf{Round 2:}
%
$\mathcal{M}_1$ receives  $\{ \lambda_{i,j}, 1 \leq i \leq t, 0 \leq j \leq s \}$. 

    $\mathcal{M}_1$ selects global boundary numbers $b_0$,$b_1$,...$b_t$. Each interval $[b_i, b_{i+1})$ is called a \textbf{bucket}.
        The selection is obtained by Algorithm \ref{alg:boundaries}.
        $b_0, b_1, ..., b_t$ are sent to all machines.

\medskip

{\bf Round 3:}
Every $\mathcal{M}_i$ sends the objects in $[b_{k-1},b_{k})$ from its local storage to $\mathcal{M}_k$, for each $1 \leq k \leq t$.
%
%
%
  Every $\mathcal{M}_i$ merges objects received in sorted order.

\hrulefill
\caption{SMMS sorting Algorithm}
\label{fig:smmsort}
\end{figure}

\subsubsection{Algorithm 1: computing bucket boundaries}
\label{sec:alg1}

Algorithm \ref{alg:boundaries} is used to compute the boundary values of
$b_0, ..., b_t$ in Round 2 of the SMMS Algorithm.
The input to this algorithm consists of the
boundary values $\lambda_{i,j}$ from each machine $\mathcal{M}_i$.
In this computation, we consider the objects
for each interval $[ \lambda_{i,j}, \lambda_{i,j+1})$ on each $\mathcal{M}_i$. Let $\mu_{i,j}$ denote the probability density distribution (pdf) in interval $[\lambda_{i,j},\lambda_{i,j+1})$. Since
 the count of data objects in the the local bucket $[ \lambda_{i,j}, \lambda_{i,j+1} )$ is $m/s$ by construction, we set
$\mu_{i,j}=({m/s})/{(\lambda_{i,j+1}-\lambda_{i,j})}$
and $\mu_{i,s}=0$, $1 \leq i \leq t$.

 Algorithm \ref{alg:boundaries} selects global boundary numbers $b_0$, $b_1$, ..., $b_t$. Each interval $[b_i, b_{i+1})$, $0 \leq i < t$ is called a \textbf{bucket}. We use the term \textit{bucket density} for the number of objects in a bucket, denoted by $\mathcal{D}[b_{k},b_{k+1})$, $0 \leq k <t$. Note that $b_k$ is not necessarily an input object, where $0 \leq k \leq t$.
The selection ensures that the estimated bucket density based on $\mu_{i,j}$, $1 \leq i \leq t, 1 \leq j \leq s$, for $\mathcal{D}[b_{k},b_{k+1})$ is equal to $m$, where $1 \leq k < t$.
 A priority queue $Q$ is maintained for storing triplets of the form $\langle \lambda,i,\mu \rangle$,
which are sorted by the first value $\lambda$ as the key.
In the triplet $\langle \lambda, i, \mu \rangle$,
$\lambda$ and $\mu$ correspond to a certain pair of
($\lambda_{i,j}$, $\mu_{i,j}$) values from $\mathcal{M}_i$.
Variable $cur$ keeps count for the estimated density of the current bucket until it reaches $m$, in which case, a new boundary $b[k]$ is determined.
Each while loop handles one sampled value $\lambda_{ij}$.
There are at most $t$ elements in $Q$, hence each while loop costs $O(\log t)$ time. The total time complexity is $O(st \log t)$ because of $t(s+1)$ rounds of the while loop.

We should point out that the complexity of Algorithm \ref{alg:boundaries} is insignificant compared to the problem size. Utilization of computer cluster is justified only when the problem size is big, and from previous works such as \cite{Tao13sigmod}, the size is in terms of billions of records and in 20 GB or more.
Thus, the value of $t$ is very small in comparison to $n$. In our experiments, the runtime for Round 2, including Algorithm \ref{alg:boundaries}, is found to be negligible for all test cases.

\begin{algorithm}[!tbp]

\SetKwInOut{input}{Input}\SetKwInOut{output}{Output}

{\small
\BlankLine
\input{$\lambda_{i,j}, \mu_{i,j}$, $1 \leq i \leq t, 0 \leq j \leq s$}
\output{Global boundaries $b[k]$, $0 \leq k \leq t$}
\BlankLine
    Initialize: Create an empty priority queue $Q$;
    $\forall 1 \leq i \leq t$: $pastpdf[i]=0$; $next[i]=0$;
     push $\langle \lambda_{i,0},i,\mu_{i,0} \rangle$ into priority queue $Q$; $pdf=0; pre=0; cur=0; k=0; flag=0$\;
\BlankLine
    \While {$Q \neq \emptyset $}{
        $\langle \lambda, i, \mu \rangle \leftarrow \text{TopAndPop}(Q)$; \ /*$\lambda$ and $\mu$ \texttt{\scriptsize from} $\mathcal{M}_i$ */\\
        \If{$flag==0$}{
            $b[k]=\lambda, k++, flag=1$; /*\texttt{\scriptsize first boundary}*/\\
        }
        \If{$(\lambda - pre ) \times pdf + cur < m$}{
            $cur+=(\lambda - pre) \times pdf$;  \  /*\texttt{\scriptsize keep count}*/\\
        }
        \Else{
            $b[k]=(m-cur)/pdf+pre, k++$; \ /* \texttt{\scriptsize new bucket} */\\
            $cur=(\lambda - pre) \times pdf+cur-m$; \ /* \texttt{\scriptsize keep count for new bucket}*/\\
        }
        $pre=\lambda$;  \ \ /*\texttt{\scriptsize update previous boundary}*/\\
         $pdf=pdf-pastpdf[i] + \mu$; \ /* \texttt{\scriptsize update pdf} */\\
          $pastpdf[i]= \mu$;  \ /*\texttt{\scriptsize pdf $\mu$ will be obsolete for $\mathcal{M}_i$}*/\\
        \BlankLine
        \If{$!next[i]$}{
            $\text{push} \langle \lambda_{i,next[i]},i,\mu_{i,next[i]} \rangle \text{into } Q, next[i]++$\;
        }
    }
    $b_t= \lambda$
    \Return $b[k], 0 \leq k \leq t$\;
}
\caption{Computing Bucket Boundaries}
\label{alg:boundaries}	
\end{algorithm}

\subsubsection{Analysis}

From our discussion in Section \ref{sec:def}, we consider total workload to be
the maximum of input size and output size. For the sorting problem, input size = output size = $n$.
The workload of a machine at a round is given by the number of objects distributed to the machines during that round. In the first round of SMMS, all machines are assigned equal workload. In Round 2, the workload is the $t(s+1)$ samples which is small compared to $n$.
Hence, we need only analyze the workload distribution at Round 3.

\begin{theorem}   \label{th:bound}
At Round 3 of SMMS sorting, the workload of each machine is bounded by
$(1+2/r+ t^2/n)m$.
\end{theorem}

\begin{proof}:
The main idea is to analyze the workload as determined by the bucket boundary.
In the first round the $m = n/t$ objects in each machine $\mathcal{M}_i$ are sorted and
divided into $s$ equi-depth (equi-frequency) intervals.
We define $[\lambda_{i,j}, \lambda_{i,j+1})$ as a local interval at machine $\mathcal{M}_i$.
Let us consider a global bucket $[b_{k-1}, b_k)$ as obtained in Round 2.
For each machine $\mathcal{M}_i$, we denote the number of objects from the local interval $[\lambda_{i,j}, \lambda_{i,j+1})$ which are in the global bucket $[b_{k-1}, b_k)$ by $a_{i,j}$ while its estimated count in $[b_{k-1}, b_k)$ is $e_{i,j}$.
Thus, $\epsilon_{i.j} =a_{i,j} - e_{i,j}$ is the error contribution of local interval $[\lambda_{i,j}, \lambda_{i,j+1})$ to the global bucket density estimation, and $\epsilon_{i} =\sum_j \epsilon_{i.j}$ is the total error contribution of machine $\mathcal{M}_i$ to the global bucket density. If $\epsilon_i$ is positive, we have an underestimation; if $\epsilon_i$ is negative, we have an over-estimation.

If objects in Rounds 1 and 2 are uniformly distributed in each interval $[\lambda_{i,j}, \lambda_{i,j+1})$, where $1 \leq i \leq t$, $0 \leq j < s$, then every $\mathcal{M}_k$'s workload in Round 3 is $m$ since each global bucket density is $m$ according to the algorithm, i.e., $\mathcal{D}[b_{k-1},b_{k})=m$, where $1 \leq k \leq t$.

Otherwise, $\mathcal{D}[b_{k-1},b_{k})=m$ does not hold, but we can bound the density $\mathcal{D}[b_{k-1},b_{k})$ as follows.
Given a global bucket $[b_{k-1}, b_k)$, where $1 \leq k \leq t$, we consider for each $\mathcal{M}_i$ the local interval $[\lambda_{i,j}, \lambda_{i,j+1})$, there are the following four cases.

\begin{figure}[!t]
\begin{center}
\includegraphics[width = 3.2in]{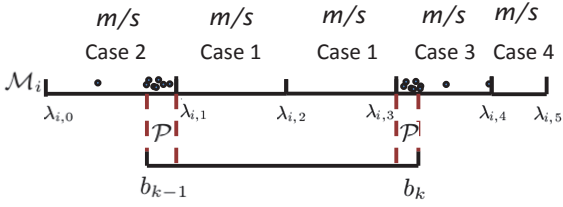}
\caption{Workload Analysis} \label{fig:sort}
\end{center}
\end{figure}

[CASE 1]: $[ \lambda_{i,j}, \lambda_{i,j+1}) \subseteq [b_{k-1},b_{k})$. In this case no error is introduced by this local interval in computing the bucket density $\mathcal{D}[b_{k-1},b_{k})$ in Algorithm \ref{alg:boundaries}, i.e., $\epsilon_{i.j}=0$ since $a_{i,j}=e_{i,j}$. For example, this is the case for $[ \lambda_{i,1}, \lambda_{i,2})$ and $[ \lambda_{i,2}, \lambda_{i,3})$ in Figure \ref{fig:sort}.

[CASE 2]: $b_{k-1} \in [ \lambda_{i,j}, \lambda_{i,j+1})$, then only some sub-interval $\mathcal{P}$ of $[ \lambda_{i,j}, \lambda_{i,j+1})$ falls in $[b_{k-1},b_{k})$. In this case error can be introduced by this local interval.
The error $\epsilon_{i.j}$ introduced by this local interval is upper bounded by $\lceil m/s \rceil$, when $a_{i,j}=m/s$, and $e_{i,j} \thickapprox 0$.  For example, such a case holds for $[\lambda_{i,0}, \lambda_{i,1})$ in Figure \ref{fig:sort}.

[CASE 3]: $b_{k} \in [ \lambda_{i,j}, \lambda_{i,j+1})$, then only some sub-interval $\mathcal{P}$ of $[ \lambda_{i,j}, \lambda_{i,j+1})$ falls in $[b_{k-1},b_{k})$. The analysis is similar to CASE 2, except that at least one object must
be located at the boundary $\lambda_{i,j+1}$ by construction. Hence the error is
given by
$\epsilon_{i.j} \leq \lceil m/s \rceil - 1$. As an example, this case holds for $[ \lambda_{i,3}, \lambda_{i,4})$ in Figure \ref{fig:sort}.

[CASE 4]: $[ \lambda_{i,j}, \lambda_{i,j+1}) \cap [b_{k-1},b_{k})= \emptyset$. In this case, $\epsilon_{i.j}=0$ since $a_{i,j}=e_{i,j}=0$. For example, see $[ \lambda_{i,4}, \lambda_{i,5})$ in Figure \ref{fig:sort}.

 Given a bucket $[b_{k-1},b_{k})$, for each $\mathcal{M}_i$, at most one of its local intervals belongs to CASE 2 and at most one of its local intervals belongs to CASE 3. Hence each $\mathcal{M}_i$ can contributes at most $2m/s + 1$ error in the estimation of $\mathcal{D}[b_{k},b_{k-1})$, i.e., $\epsilon_{i}=2m/s + 1$.
 Therefore, for each bucket $[b_{k-1},b_{k})$, the total error can be at most $(2m/s + 1)*t=2m/r + t$ (since $s = rt$), because of $t$ machines and the fact that the estimated bucket density is $m$. Hence, the actual bucket density is at most $m + 2m/r + t$ = $(1+2/r+t/m)m$. Thus, in Round 3 of SMMS, the workload of every machine $\mathcal{M}_i$, $1 \leq i \leq t$, is upper-bounded by $(1 + 2/r + t^2/n) m $. $\qed$
\end{proof}

Theorem \ref{th:bound} gives a bound for the worst case workload distribution.
We choose $r$ to be a small integer.
For example, if $n \geq 25M$, $r = 2$, and $t=50$, then the workload for each machine is bounded above by $\approx 2m$.
If $n \geq 75M$, $r=6$, and $t=50$, then this
bound becomes $\approx 1.3 m$.
To our knowledge, SMMS has the best theoretical bound for workload balancing among known sorting algorithms.

\begin{theorem}
Given $n$ objects and $t$ machines,
SMMS sorting is $(3,(1+2/r+ rt^3/n))$-minimal given $t^3 \leq n$.
\label{th:smms}
\end{theorem}

\begin{proof}:
Firstly, SMMS consists of 3 rounds when implemented in MPI.
Let $k = (1 + 2/r + rt^3/n)$.
The comparable sequential algorithm $\mathcal{A}_{seq}$ is taken as the external merge sort algorithm on the given data set on a single machine, since SMMS sorts and merges the objects.
In Inequality (\ref{ineq:workload}), $W_{seq} = n = mt$.
From Theorem \ref{th:bound}, since $r$ is a positive integer, Inequality (\ref{ineq:workload})
for workload distribution holds for SMMS.

For Inequality (\ref{ineq:network}),  $N$ = $2n = 2mt$, since $n$ is both the input size and the output size.
For Round 2, $M_1$ receives $\beta = rt^2$ objects, and sends
$\gamma = t^2$ values.
Since $m = n/t$, $\beta = (rt^3/n)m$ and $\gamma = (t^3/n)m$.
Thus, $\beta + \gamma \leq (rt^3/n)2m$.
From Theorem \ref{th:bound}, in Round 3, each machine $\mathcal{M}_i$, $1 \leq i \leq t$, sends at most $m$ objects and receives at most
$(1+2/r+t^2/n)m$ objects.
The total objects sent and received by $M_i$ in Round 3 is given by
$(2 + 2/r + t^2/n)m < (1 + 2/r + t^2/n)2m \leq (1 + 2/r + rt^3/n)N/t$.
Hence, Inequality (\ref{ineq:network}) holds for SMMS.

The computation time includes that of sorting, merging and the bucket boundary computation.
The cost is dominated by the sorting process. On each machine the time taken for the sorting at Round 1 is given by
$O(m \log m)$ = $O((n/t)(\log n - \log t))$
= $O((n\log n )/t)$.
Note that $\mathcal{A}_{seq}$ takes $O(n \log n)$ time.
For the bucket boundary computation, the cost at $\mathcal{M}_1$
is given by $\beta = O(st \log t) = O(rt^2 \log t)$.
Since $r$ is a small constant, if $t^3 < n$, $\beta$ = $O((n \log n )/t)$.
Hence, Equality (\ref{ineq:cpu}) holds.
We derive that SMMS is $(3,(1+2/r+rt^3/n))$-minimal given $t^3 \leq n$. $\qed$
\end{proof}

We choose $r$ to be a small integer.
For example, if $n \geq 25M$, $r = 2$, and $t=50$, then from Theorem \ref{th:smms}, SMMS is
(3,2.01)-minimal.
If $n \geq 75M$, $r=6$, and $t=50$, then SMMS is (3,1.35)-minimal.
These show that SMMS is highly effective in the parallelization of the computation. The performance guarantee of SMMS is confirmed by our empirical study where we show close to perfect workload distributions and results of speedup that follow a near-linear increase with the number of processors.

%

\subsection{Terasort - a randomized algorithm}

Terasort is a parallel algorithm proposed to sort data in the size range of terabytes \cite{Malley08Yahoo}.
There are 3 rounds in Terasort:
(1) a random sample set is collected from the input.
(2) From the sample set, range boundaries are determined for $t$ contiguous but disjoint
ranges that partition the data according to the sort key values.
(3) The objects that fall into a particular range are
sent to a corresponding machine.
Each machine will then sort the objects received so that
combining the results of all machines gives a sorted result of the given dataset.
The pseudo code for Terasort is given in Figure \ref{Terasort}.
We have implemented Terasort in MPI and it takes 3 rounds as described.
Note that if implemented in MapReduce, then Rounds 1 and 2 form one MapReduce job,
and Step 3 forms another MapReduce job.

\begin{figure}[htbp]
\hrulefill
\\
{\texttt{\bf Terasort}}

\vspace*{-2mm}
\hrulefill

{\bf Round 1 :}
Every machine $\mathcal{M}_i$, $1 \leq i \leq t$, samples each object from its local storage with probability $\rho$
independently. The samples are sent to $\mathcal{M}_1$.

\medskip

{\bf Round 2 :} 
Let $S_{samp}$ be the set of samples received by $\mathcal{M}_1$, and $s = |S_{samp}|$.
$\mathcal{M}_1$ sorts $S_{samp}$ and picks $b_1, ..., b_{t-1}$ where
$b_i$ is the $\lceil i \times s/t \rceil$-th smallest object in $S_{samp}$,
for $1 \leq i \leq t-1$. Each $b_i$ is a \textbf{boundary object}. $\mathcal{M}_1$ sends $b_1, ..., b_{t-1}$ to each other machine.

\medskip

{\bf Round 3 :}
Assume $b_1, ..., b_{t-1}$ have been sent to all machines.
Every $\mathcal{M}_i$ sends the objects in $(b_{j-1}, b_j]$ from its local
storage to $\mathcal{M}_j$, for each $1 \leq j \leq t$, where
$b_0 = -\infty$ and $b_t = \infty$ are dummy boundary objects.


Every $\mathcal{M}_i$ sorts the objects received

\hrulefill
\caption{Terasort}
\label{Terasort}
\end{figure}

Let $\mathcal{S}_i = \mathcal{S} \cap (b_{i-1}, b_i]$, for $1 \leq i \leq t$.
In Round 3, $\mathcal{S}_i$ is collected by $\mathcal{M}_i$ and they are sorted.
An interesting and useful result is derived in \cite{Tao13sigmod}
showing that if the sample probability $\rho$ is set to $1/m \ln(nt)$,
then with high probability,
the number of objects distributed to each machine is $O(m)$.
Our objective is $(\alpha,k)$-minimality for a small $k$. Hence, we aim to ensure that
the load distribution, $|\mathcal{S}_i|$, is bounded by $km$ for a small $k$.

We first make some change to the above algorithm in the
randomization step (Round 1).
We replace the step of sampling each object with probability $\rho = \ln(nt)/m$
by Algorithm $\mathbb{S}$ below. Algorithm $\mathbb{S}$ will always return exactly
$\lceil \ln (nt) \rceil $ objects.

\vspace*{2mm}

\hspace*{-4mm}{\textbf{ Algorithm} $\mathbb{S}$:}
Given objects $o_1, ..., o_m$, initially no object is selected. Next
consider objects one by one from $o_1$ to $o_m$,
when considering object $o_k$, let $j$ be number of
objects already selected, select object $o_k$ with probability
$(\lceil \ln(nt) \rceil - j)/(m-k+1)$.

\medskip

We have the following lemma from \cite{Fan62JASA,Jones62cacm,Knuth97}.

\begin{lemma}
With Algorithm $\mathbb{S}$, exactly
$\lceil \ln(nt) \rceil$ objects will be selected by $\mathcal{M}_i$ and
the sampling is completely unbiased so that
the probability of selecting any given object is
$\lceil \ln(nt) \rceil /m$.
\label{lem:AlgS}
\end{lemma}

From Lemma \ref{lem:AlgS}, the size of $s = |S_{samp}| = t \lceil \ln(nt) \rceil$.
We shall prove that Terasort with this sampling algorithm
approaches close to 5-linear speedup with high probability.
Our proof makes use of the Chernoff bound and properties of a set of sliding
buckets constructed out of the list of sorted objects.
First we describe the Chernoff bound to be used.

Let $X_1, ..., X_n$ be independent Bernoulli variables with
$Pr[X_i = 1] = p_i$, for $1 \leq i \leq n$.
Set $X = \sum_{i=1}^n X_i$ and $\mu = E[X] = \sum_{i=1}^n p_i$.
The Chernoff bound states that
for any $0 < \alpha < 1$,
$Pr[X \leq (1- \alpha)\mu] \leq exp(-\alpha^2 \mu/2)$.

\begin{figure}[!t]
\begin{center}
\includegraphics[width = 1.8in]{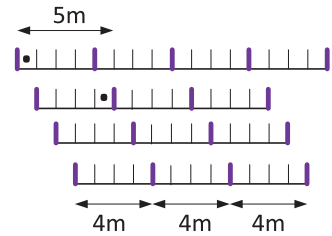}
\caption{4 groups of sliding buckets with $4m$ objects per bucket, $t=16$, $|S| = 4(4m)$, $h = 4$, $k=0$} \label{fig:sliding}
\end{center}
\end{figure}

\begin{theorem}
Given $n \geq 4t$ as input size for Terasort with Algorithm $\mathbb{S}$,
$|S_i| \leq 5m+1$ with probability at least $1 - 1/n$.

\label{thm:tersortwork}
\end{theorem}

\begin{proof}:
Given $|S| = n$ and $m = n/t$.
%
%
Let $S_{samp}$ be the set of samples received by $M_1$.
 From Lemma \ref{lem:AlgS}, $s = |S_{samp}| = t\lceil \ln (nt) \rceil$.


Imagine that $S$ has been sorted in ascending order.
From the sorted list we form overlapping sub-lists or \textbf{buckets}
in a sliding window manner. (Note that the buckets defined here are different from the buckets defined for SMMS.)
Since $|S| = n = mt$,
let $|S| = n = 4hm+k$, where $h,k$ are integers and $0 \leq k < 4$.
Each bucket has $4m$ up to $4m+1$ objects.
We form 4 groups of buckets.
The first group is a partition of $S$ forming $\lfloor t/4 \rfloor$ buckets, so that each of the first $k$ buckets contains
$4m+1$ objects, and each of the remaining buckets contains $4m$ objects.
The $j$-th group of buckets are formed beginning with the $((j-1)m+1)$-th object in $S$,
so that the first $k-j$ buckets contain $4m+1$ objects each and the remaining buckets contain $4m$ objects each. There are $t-3$ buckets in total.
If we sort the buckets according to their smallest object, and call the resulting ordered buckets
bucket 1, ..., bucket $t-3$, then
a bucket and its next bucket overlap by either $3m$ or $3m+1$ objects.
Figure \ref{fig:sliding} shows the scenario where $|S| = 4(4m)$, i.e. $h=4$ and $k = 0$.

$S_i$ is defined to be between two boundary objects. Suppose that each bucket contains at least one boundary object, then the furthest distance between two consecutive boundary objects is found when a boundary object is the smallest object in bucket $j$ and another boundary object is the greatest object of bucket $j+1$, for $1 \leq j \leq t-3$.
Hence $|S_i| \leq 5m+1$. Next we determine the condition for each bucket to contain at least one boundary object.

A bucket $\beta$ definitely includes a boundary object if $\beta$ covers
more than $\lceil \ln (nt) \rceil = s/t$ samples, as one boundary object is taken every
$\lceil s/t \rceil$ consecutive samples.
Let $|\beta|$ be the number of objects in $\beta$.
$|\beta| \geq 4m$.

Define random variables $x_j$, $1 \leq j \leq |\beta|$,
to be 1 if the $j$-th object in $\beta$ is sampled, and 0 otherwise. Let
\[
\mbox{$X = \sum_{j=1}^{|\beta|} x_j = |\beta \cap S_{samp}|$}
\]

Clearly, $E[X] \geq  4m \rho = 4 \lceil \ln (nt) \rceil$. Thus,
\begin{eqnarray*}
& &Pr[X < \lceil \ln(nt) \rceil ] \\
&=& Pr[X < (1-0.75) 4 \lceil \ln (nt) \rceil ]\\
&\leq& Pr[X \leq (1 - 0.75) E[X]]\\
&\leq&  \mbox{$exp \left( - (0.75)^2 \frac{E[X]}{2} \right)$} \ \ \hspace*{3mm}  \mbox{  ... by Chernoff Bound}\\
&\leq& \mbox{$exp \left( - 0.5625 \frac{4 \lceil \ln (nt) \rceil}{2} \right)$}
\leq exp( -\ln (nt) )
\leq 1/(nt)
\end{eqnarray*}

We say that a bucket fails if it covers no boundary object.
The above shows that a bucket fails with probability at most $1/(nt)$.
There are $t-3$ buckets.
As in the construction of the buckets,
we can partition the $t-3$ buckets into 4 groups so that there are
$t/4$ buckets in the first group and $t/4-1$ buckets in the remaining groups.
In each group, all the buckets are non-overlapping, and the probability of one or more of these buckets fail is bounded by $t/4 \times 1/(nt) = 1/(4n)$ or $(t/4-1) \times 1/(nt) < 1/(4n)$.
By union bound, the probability that one or more buckets fail overall
is upper bounded by $4 * 1/(4n) = 1/n$. $\qed$
\end{proof}

\begin{corollary}
The workload $W_i$ at each round of Terasort is upper-bounded by $(5+t/n)W_{seq}/t$
with probability at least $1-1/n$.
\label{cor:terasort}
\end{corollary}

\begin{theorem}
Given $n$ objects and $t$ machines, assuming $ln(nt)<t$,
Terasort with Algorithm $\mathbb{S}$ is $(3,(5+t^3/n))$-minimal with a probability of $1-1/n$.
\end{theorem}

\begin{proof}:
Firstly, there are 3 rounds when Terasort is implemented in MPI.
The comparable sequential algorithm $\mathcal{A}_{seq}$ is the external merge sort algorithm on the given data set
on a single machine.
From Theorem \ref{thm:tersortwork}, since $5m+1 = (5+1/m)m = (5+t/n)m < (5+t^3/n)m$, the workload distribution of Terasort satisfies Inequality (\ref{ineq:workload}).
Next consider Inequality (\ref{ineq:network}).
For Round 2, $\mathcal{M}_1$ receives $\beta = t\ln (nt)$ objects, and sends $\gamma = t^2$ objects.
Since $m = n/t$. We have $\beta = (t^2 \ln(nt)/n)m$ and $\gamma = (t^3/n)m$, thus, $\beta + \gamma < (t^3/n)2m$ if $ln(nt)<t$.
For Round 3, each machine sends at most $m$ objects and receives at most
$5m+1$ objects.
Hence, Inequality (\ref{ineq:network}) is satisfied.
The computation time includes that of sorting, merging and the bucket boundary computation.
The cost is dominated by the sorting process. On each machine the time taken for this sorting is
$O(m \log m) = O( n/t \log n/t )$. The computation time for $\mathcal{A}_{seq}$ is
$O(n \log n)$.
Hence, Equality (\ref{ineq:cpu}) holds. $\qed$
\end{proof}



\subsection{Discussion}
\label{sec:sortdiscuss}

So far we have considered sorting a set of objects.
The problem is more complex for sorting a bag of objects in which some objects may have the same key. To deal with this, for SMMS, after the first round where
each machine sorts its portion of data, objects with the same key will be assigned a special object key type. The object key will consist of the machine id so that such keys assigned to objects of the same original key in all machines are unique. With the object keys, we deal with the bag of objects as a set of objects.

In the previous subsections we showed that our proposed method of SMMS and Terasort are $(3,k)$-minimal for some small $k$ values. As in \cite{Tao13sigmod}, we can extend these properties to other problems that use sorting as a major step, including the problems of ranking, skyline, group by queries, semi-join, and sliding aggregation.

Comparing Theorem \ref{th:bound} and Corollary \ref{cor:terasort}, SMMS has a better theoretical guarantee for workload balancing.
From the results of $(\alpha,k)$-minimality for SMMS and Terasort, SMMS enjoys a smaller $k$ value. Another advantage of SMMS is that if we only allow internal sorting, then the RAM requirement of each machine is to hold $m$ objects, which is needed for the first round when data is distributed evenly to all machines.
For the second round, we only need storage to hold $t^2$ objects. For Round 3, each machine merges sorted data objects from all machines, and for that we can use a priority queue of size $t$. The sorted objects from all machines are
entered into the priority queue in sorted order. Whenever the queue is full, the minimum value is deleted from the queue and output to disk. The main memory required for this process is $O(t)$. Since $t << m$, overall, we only need storage to hold $m$ objects at each machine. There is no comparable RAM storage bound for Terasort, which requires much more RAM storage for a similar guarantee.

%
%
\section{Skew Join}  \label{sec:join}

We consider the equi-join of two tables $S$ and $T$.
Due to its importance, Hadoop offers a standard solution, \textbf{Standard Repartition Join}\footnote{This solution is found in the
package org.apache.hadoop.contrib.utils.join.} \cite{Blanas10sigmod}. This is a MapReduce algorithm.
In the map phase, each map task works on a split of $S$ or $T$.
Each tuple is tagged with the table name $S$ or $T$.
The extracted join key and the tagged tuple are output as
a $(k, v)$ pair. 
In the shuffling phase, all tuples for each join key value are input to
a reducer. The reducer separates the tuples into two sets, one from each table, by means of the table tag.
Then a cross-product operation is carried out over the two sets and the result is returned as part of the answer.
A major problem with this method is that it cannot handle skew data.
If some join key
value appears in a large number of tuples in both $S$ and $T$, then the
join result for this key will be very large, and the workload for the machine handling this key will be excessively heavy compared to other
machines. This greatly affects the overall speedup.
In this section we study this problem of Join Product Skew (JPS).


\subsection{Preliminaries}

We consider the problem of joining two tables $S$ and $T$ with an equality join condition of $S.\rho = T.\rho$ for a certain join key $\rho$.
As in \cite{Okcan11sigmod},
we model the join result by means of a $|S| \times |T|$ \textbf{join-matrix} $\Gamma$ as shown in
Figure \ref{fig:join}(b). In this matrix, $S$ and $T$ are sorted by the join key into ordered
lists $\overrightarrow{S} =$
$s_1, s_2, ..., s_{|S|}$, and $\overrightarrow{T} = $ $t_1, t_2, ..., t_{|T|}$.
In Figure \ref{fig:join}(b), the key values for $s_1,...,s_{|S|}$ are
$b,d,d,d,d,f$, correspondingly.
The matrix entry $\Gamma(i,j)$ is true (shaded) iff $s_i.\rho = t_j.\rho$.
The join result for a certain join key $k$ form a shaded rectangular region in $\Gamma$,
we call this region the \textbf{join result} for $k$, or simply
{\emph{result}($k$)}.
For example, in Figure \ref{fig:join}(b), the join result for key $d$, denoted by $result($d$)$, is the
shaded rectangle of size $4*3$.

Suppose $k$ is a join key, we say that the \textbf{size of the join result} for $k$ is $M \times N$ if $M$ and $N$ are the number of tuples with key $k$ from $S$ and $T$, respectively.
 For example, in Figure \ref{fig:join}(b), the join result for key $d$ has size $4 \times 3$, which is the cross product of tuples $2$ to $5$ from $S$ and $2$ to $4$ from $T$.
Next we define the skew factor to indicate how large
 the join result size is compared with the total size of $S$ and $T$, where size is measured by the number of tuples.

\begin{definition}[Join Skew Factor $\sigma$]
The skew factor of the join, $S \Join T$, of two tables $S$ and $T$ is
given by $\sigma$ if $| S \Join T | =
\sigma(|S|+|T|)$.
\end{definition}

\subsection{RandJoin- A Randomized Algorithm}

In this subsection, we introduce our randomized algorithm for
handling skew join. We call our algorithm \textbf{RandJoin}.

\begin{figure}[!tbp]
\begin{center}
\hspace{-5mm}
\includegraphics[width = 1.9in]{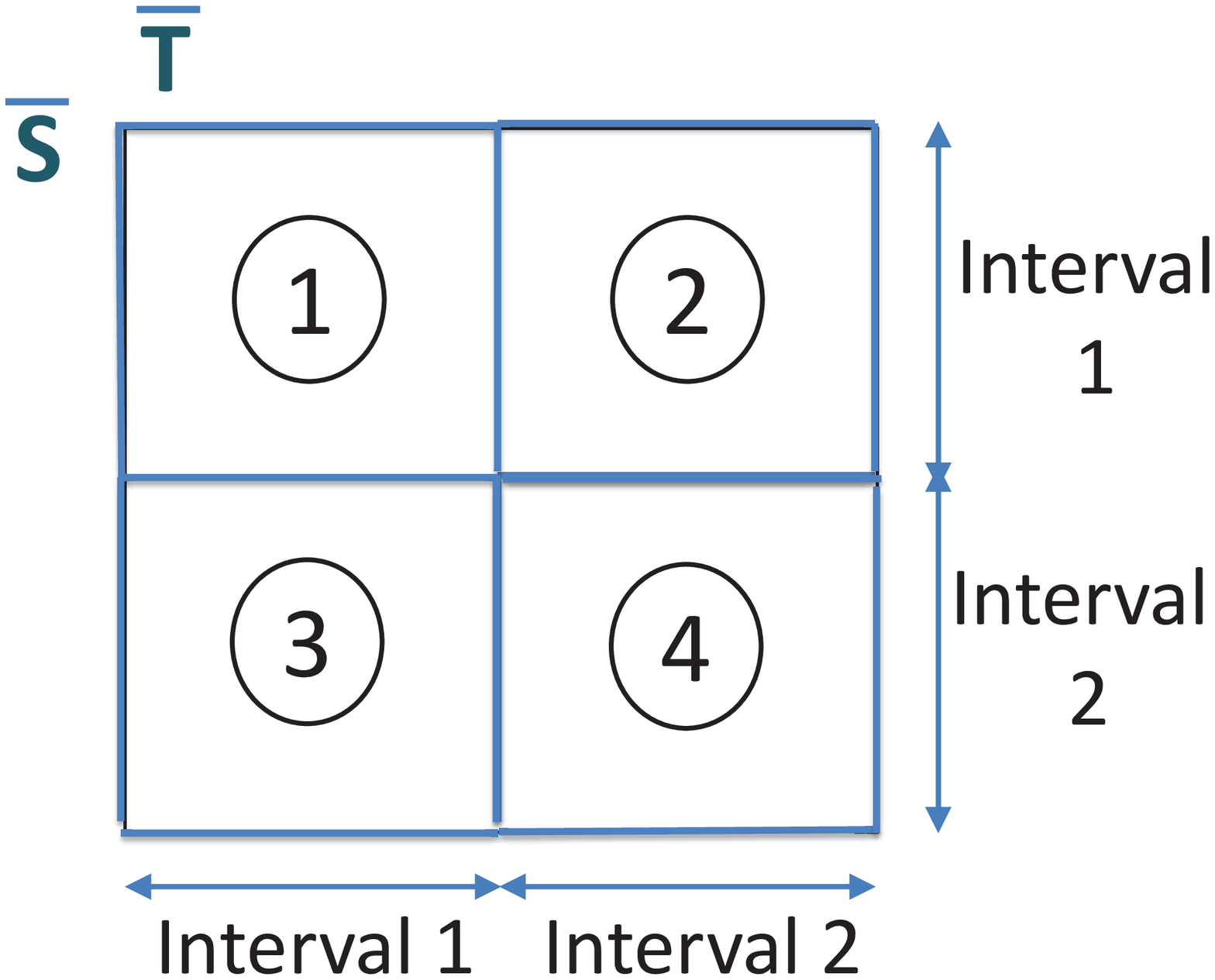}\hspace*{-3mm}
\includegraphics[width = 1.6in]{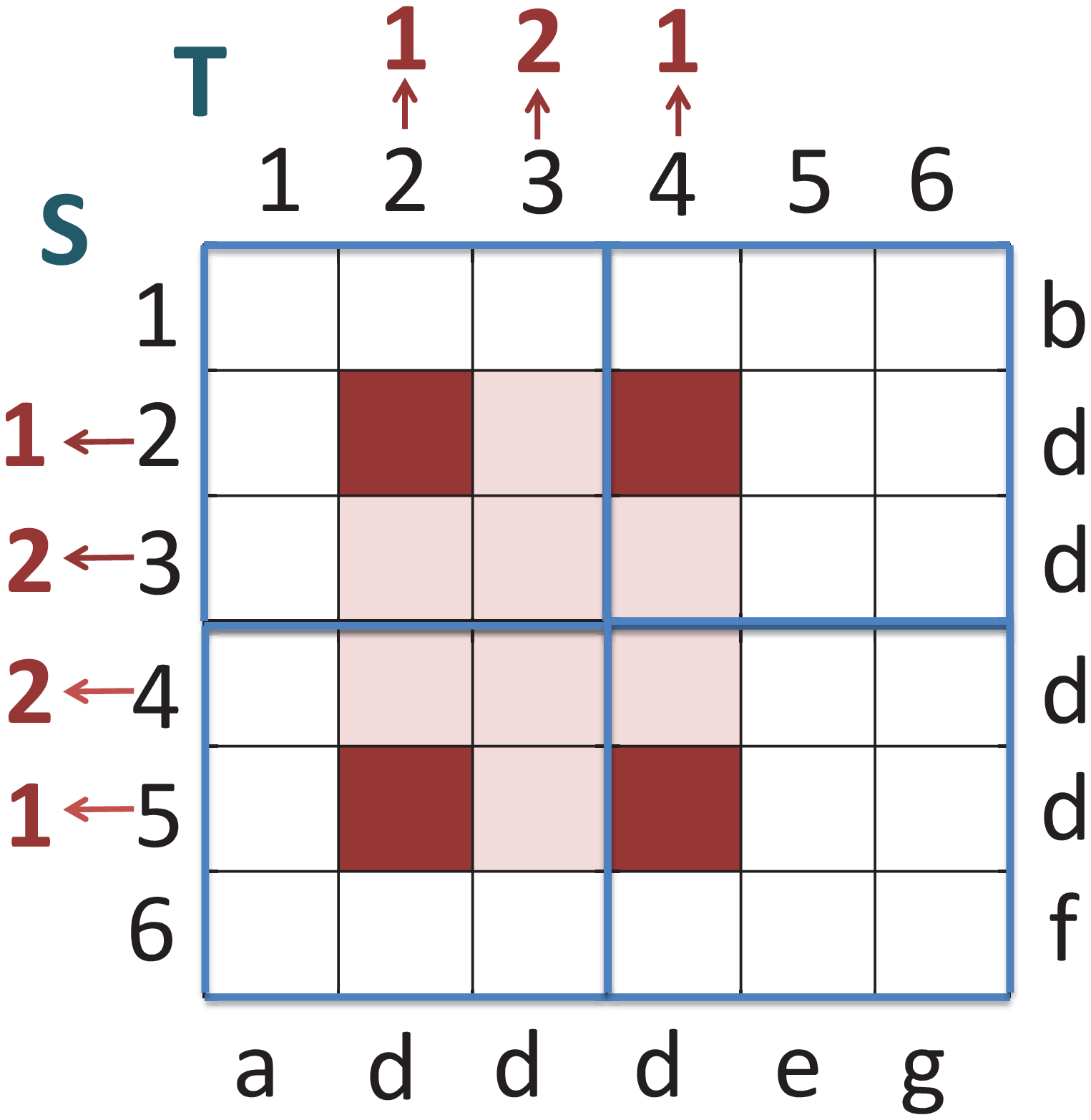}\\
\hspace*{3mm}(a)\hspace*{45mm}(b)
\caption{(a) machine matrix $A$ for 4 machines ($t = 4$), $a=b=2$. (b) join matrix $\Gamma$ and
randomized tuple-to-interval mapping} \label{fig:join}
\end{center}
\end{figure}

\subsubsection{Machine Matrix $A$}
Let the number of machines be $t$, we determine two integers
$a$ and $b$ such that firstly,
$a \times b = t$ and secondly, among all $a,b$ satisfying $a \times b = t$,  $a|T| + b|S|$ is minimized.
We shall see that $a \times b = t$ is a sufficient condition for our workload balancing guarantee. The minimization of $a|T|+b|S|$ can lead to some minor improvement for load balancing related to the join input size to each reducer.
The reason for this choice will be explained later.
%
With the values of $a$ and $b$, we form a
$a \times b$ matrix $A$ called the \textbf{machine matrix}.
For matrix $A$,
we call the first dimension $\overline{S}$ and the second dimension $\overline{T}$.
Each $A[i,j]$ is assigned a unique machine.
We say that $A[i,j]$ lies on \textbf{interval} $i$ of $\overline{S}$
and \textbf{interval} $j$ of $\overline{T}$.



\begin{example}
Fig.\ref{fig:join}(a) shows the machine matrix $A$ given 4 machines.
The two dimensions of $\overline{S}$ and $\overline{T}$ each consists of 2 intervals. That is, $a$ and $b$ are both 2.
The machines are assigned to the matrix elements so that
machines $\mathcal{M}_1$, $\mathcal{M}_2$, $\mathcal{M}_3$, and $\mathcal{M}_4$ are assigned to
$A[1,1]$,
$A[1,2]$,
$A[2,1]$,
and $A[2,2]$, respectively.
\end{example}

\subsubsection{Tuple-to-Interval Mapping}
We assign tuples to machines by a randomized algorithm.
For each tuple in $S$ we randomly select an integer $i$ in
$1,...,a$ and map the tuple to interval $i$ of $\overline{S}$ in the machine matrix $A$. For each tuple in $T$, We randomly select an integer $j$ in $1,...,b$ and map the tuple to interval $j$ of $\overline{T}$ in $A$. Then each tuple is assigned to the machines as follows:
if an $S$ tuple $x$ is mapped to interval $i$ of $\overline{S}$ in matrix $A$, then $x$ is
sent to each of the $b$ machines assigned to $A[i,1]$,
$A[i,2]$, ..., $A[i,b]$.
If a $T$ tuple $y$ is assigned to interval $j$ of $\overline{T}$ in $A$, then
$y$ is sent to each of the $a$ machines assigned to
$A[1,j]$, $A[2,j]$, ..., $A[a,j]$.
Each machine computes the cross-product of all the
$S$ tuples and $T$ tuples that it has received.
Hence the join result for tuples $x$ and $y$, if any, will be uniquely generated by the machine assigned to
$A[i,j]$.



\begin{example}:
In Figure \ref{fig:join}(b), we show the join matrix for the tables $S$ and $T$. Each table contains 6 tuples. We show that tuples 2,3,4,5 of $S$ are randomly assigned
interval numbers 1,2,2,1. Then the second tuple of $S$ will be mapped to the first interval on
$\overline{S}$ in matrix $A$ in Figure \ref{fig:join}(a), and it will be sent to machines $\mathcal{M}_1$ and $\mathcal{M}_2$.
The join result in the join matrix for the darker shaded area will be generated by machine $\mathcal{M}_1$.
\end{example}

From the above tuple-to-interval mapping, each tuple in $S$ is
assigned to $b$ machines, and each tuple in $T$ is assigned to $a$ machines. By selecting $a$ and $b$ that minimize $a|T|+b|S|$ we minimize the total input size to the machines in the number of tuples. Note that this optimization has no impact on the output size, which is dominating.

Note also that the assignment of a tuple to multiple machines is necessary to distribute the workload. Assume on the contrary that each tuple is only assigned to one machine. Then, in Figure \ref{fig:join}, all the tuples in the join result must be assigned to a single machine. For skew data, this will result in highly unbalanced workload, similar to the problem found in the Standard Repartition Join algorithm in the standard Hadoop package.


We have implemented RandJoin in Hadoop MapReduce.
There is only one MapReduce round.
The MapReduce program is shown in Figure \ref{mapReduceSkew}.

\begin{figure}[htbp]
\hrulefill
\\
{\texttt{\bf RandJoin}}

\vspace*{-2mm}
\hrulefill

Given $t$ machines, we determine integers
$a$ and $b$ such that $ab = t$.
The $a \times b$ matrix $A$ is formed and machines are
assigned to $A[i,j]$ for $1 \leq i \leq a$, $1 \leq j \leq b$.

\medskip

{\emph{Map Phase:}}
\vspace*{-2mm}
\begin{itemize}[leftmargin=.1in]
\item[]
$\mathcal{M}_i$ reads the values of pairs of the form
$((k, S), v)$ or $((k, T), v)$
where $k$ is the join key value,
$S$ and $T$ are the table ids, $v$ is the payload.
Note that $(k,S)$ and $(k,T)$ are composite keys.
\item[]
For a tuple from $S$ of the form $((k, S), v)$,
$\mathcal{M}_i$ randomly select an integer $i$ in $[1,a]$,
and we map this tuple to interval $i$ of $\overline{S}$.
Next we send the tuple to the machines assigned to $A[i,1],...,A[i,b]$.
Similarly for each $T$ tuple, we randomly select
an integer $j$ in $[1,b]$, map the tuple to interval $j$ of $\overline{T}$, and send the tuple
to the corresponding machines.
\end{itemize}

{\emph{Reduce Phase:}}
\begin{itemize}[leftmargin=.1in]
\item[]
Each machine $\mathcal{M}_i$ receives tuples from the map phase, and join results are generated by a
cross-product of tuples of the same key from $S$ and from $T$.
\end{itemize}
\hrulefill
\caption{MapReduce algorithm for RandJoin}
\label{mapReduceSkew}
\end{figure}

%

\subsubsection{Analysis of RandJoin}

Next we analyse the workload distribution under RandJoin.
Consider a join result of size $M \times N$.
Let $X_{i_1}$, $X_{i_2}$,...$X_{i_M}$ be random variables, such that $X_{i_k}=1$ if the $i_k$-th tuple from $S$ is assigned to interval $i$ of $\overline{S}$ and $X=\sum_k X_{i_k}$. Analogously, let $Y_{j_1}$, $Y_{j_2}$,...$Y_{j_N}$ be random variables, such that $Y_{j_k}=1$ if the $j_k$-th tuple from $T$ is assigned to interval $j$ of $\overline{T}$ and $Y=\sum_k Y_{j_k}$.
Hence $X$ ($Y$) is the number of tuples assigned to an interval on $\overline{S}$ $ \left( \overline{T} \right)$.
Note that both $X$ and $Y$ follow binomial distribution since each row and each column of the matrix $A$ is assigned uniformly at random to tuples (with random selections from $[1,a]$ and $[1,b]$, respectively), while $X$ and $Y$ are independent with each other.
\begin{eqnarray*}
\mbox{$X \sim B(M,1/a)$}; \ \
\mbox{$Y \sim B(N,1/b)$}\\
\mbox{$E[X]=\mu_x = M/a; \ \
E[Y]=\mu_y = N/b$}
\end{eqnarray*}

  Suppose $X$ tuples out of $S$ are mapped to interval $i$ and $Y$ tuples out of $T$ are mapped to interval $j$, and
 machine $\mathcal{M}_r$ is assigned to $A[i,j]$, then $X \times Y$ out of the join result size of $M \times N$ are assigned to $\mathcal{M}_r$. For example, in Fig.\ref{fig:join} (b), $2$ out of $4$ key $d$ tuples in $S$ and $2$ out $3$ key $d$ tuples in $T$ are assigned to machine $\mathcal{M}_1$, hence $2 \times 2=4$ of the join result tuples are assigned to $\mathcal{M}_1$.

\begin{eqnarray}\mbox{$E[XY]=E[X]E[Y]=\frac{M N}{a b}$}
=\mbox{$(MN)/t = \mu_x  \mu_y$}
\\
\mbox{$Pr[X > (1 + \delta) \mu_x] < \left(\frac{e^\delta}{ (1+\delta)^{1+\delta}} \right)^{\mu_x}$}
\label{chernoff}
\end{eqnarray}

We have the above probability inequality according to Chernoff Bound.
  We shall apply this inequality to derive a bound for the output size for each machine.
\begin{eqnarray}
\mbox{ Let
$f(\delta) = \left(\frac{e^\delta}{ (1+\delta)^{1+\delta}} \right)$ \hspace*{9mm}}
\label{eqn:f(delta)}
\\
\mbox{Then we have $Pr[X > (1 + \delta) \mu_x] < \left( f(\delta) \right)^{\mu_x}$}
\label{chernoff3}
\\
\mbox{Similarly, $
Pr[Y > (1 + \delta) \mu_y] < \left( f(\delta) \right)^{\mu_y}$}.
\label{chernoff4}
\end{eqnarray}

\begin{lemma}Given $f(\delta)$ as defined in Equation \ref{eqn:f(delta)},
$$\mbox{$Pr[XY < ( 1 + \delta)^2 M N/t] > 1 -  \left( f(\delta) \right)^{\mu_x} - \left(f(\delta) \right)^{\mu_y}$}$$
\label{lem1}
\end{lemma}

\begin{proof}:
If $X < (1 + \delta) \mu_x$ and $Y < (1 + \delta) \mu_y$
then $XY < (1 + \delta)^2 \mu_x\mu_y$ or $XY < ( 1 + \delta)^2 M N/t$.
By Union Bound, the probability that $X > (1+\delta) \mu_x$ or
$Y > (1+\delta) \mu_y$ is bounded
above by $Pr[X > (1+\delta) \mu_x] + Pr[Y > (1+\delta) \mu_x]$.
Hence the lemma follows from Equations \ref{chernoff3} and \ref{chernoff4}. $\qed$
\end{proof}

\begin{corollary}
If $M/a \geq 300$ and $N/b \geq 300$
then the probability that
$XY < 2 MN/t$ is greater than $1 - 1.2 \times 10^{-9}$.
\label{cor1}
\end{corollary}

\begin{proof}:
Let us set $\delta = 0.4$,
then $f(\delta) < 0.9315$.
If $\mu_x = 300$, then $\left(0.9315\right)^{\mu_x} < 5.7\times 10^{-10} $.
If $\mu_y = 300$, then $\left(0.9315\right)^{\mu_y} < 5.7\times 10^{-10} $.
%
%
%
By Lemma \ref{lem1}, the probability that
$XY < (1.4)^2 M N/t < 2 MN/t$ is greater than $1 - 2 \times 5.7 \times 10^{-10} > 1 - 1.2 \times 10^{-9}$. $\qed$
\end{proof}

Suppose we have 100 machines, Corollary \ref{cor1} says that given that the join result size for a key is at least $M \times N = 3000 \times 3000$, then
the probability that the join result assigned for any machine is no more than
twice the even distribution size is at least $ 1 - 1.2\times 10^{-9}$.

The above result is for the join result for one key only. We can have join results for more keys. If all join results satisfy the condition stated in Corollary \ref{cor1}, since they
are independent in the randomization process,
the workload distribution guarantee also holds.

To our knowledge, the above results are the best known workload distribution guarantees for skew join algorithms compared with previous works.

\begin{corollary}
If the join results for each join key is either an empty set or a set with size $M \times N$ where $M/a \geq 300$ and $N/b \geq 300$, then the probability that the workload of any machine is less than twice the even workload is
more than $1 - 1.2 \times 10^{-9}$.
\end{corollary}

\begin{theorem}
 RandJoin is $(1,(2+t/\sigma))$-minimal with a probability of  $1 - 1.2 \times 10^{-9}$ if each non-empty join result for a join key has size $M \times N$ where $M/a \geq 300$, $N/b \geq 300$, for integers $a,b$, s.t. $ab$ is the number of machines, and $\sigma$ is the skew factor.
\label{thm:randjoin}
\end{theorem}

In the above the term $t/\sigma$ covers the input size of a machine, which is bounded by $|S|+|T| \leq 1/\sigma(|S \Join T|)$. Since $n=(|S|+|T|+|S \Join T|)>=(\sigma+1)(|S|+|T|),
thus |S|+|T| <=(mt)/(\sigma+1)$.
Since the skew factor $\sigma$ is typically very large, RandJoin is approximately $(1,2)$-minimal.

%

\subsection{StatJoin - A Deterministic Algorithm}


In this section we introduce a deterministic algorithm \textbf{StatJoin} for handling the skew join problem
The major idea for StatJoin is the partitioning of data based on \emph{statistical} information.

\subsubsection{Statistics Collection}

In Algorithm StatJoin, we first collect statistics from the two tables $S$ and $T$.
For this purpose, we apply a parallel sorting algorithm such
as Terasort or SMMS for each of $S$ and $T$, allowing for repeated keys. After sorting, each
$\mathcal{M}_i$ contains
sorted portions or buckets $\mathcal{P}^S_i$ and
$\mathcal{P}^T_i$ of $S$ and $T$, respectively.
All occurrences of the same join key will be collected at one single machine.
Then each machine calculates
the sizes of the join results
for different join keys, and the total join result size that will be generated from
$\mathcal{P}^S_i$ and
$\mathcal{P}^T_i$. The result sizes are measured in number of tuples.
%
%
%
Based on such statistics, a task distribution algorithm is
applied on all the join tasks.

Let $W$ be the total join result size.
A join result of a key with a size greater than $W/t$ is called a
{\bf big join result}, otherwise, it is called a {\bf small join result}. Note that the biggest size of a small join result is $W/t$.
We decide on the task distribution by first considering the
big join results, followed by the consideration of the small join results.

Although the statistics collection requires a sorting of the input datasets, the overhead for this computation is insignificant when compared to the overall runtime. From our experiments, the overall runtime is no more than that of the RandJoin algorithm which does not require any statistics collection step. There are two reasons for these results. The first reason is that sorting of the input is not costly when compared to the join operation, because the input size is very small when compared to the result size. The second reason is that in the MapReduce process, the shuffling step sorts the key-value pairs and the sorting is sensitive to the original sorted ordering of the keys. The sorting in StatJoin leads to a more efficient shuffling step when compared to RandJoin. In other words, the sorting in StatJoin is a useful computation for the later join step.

\subsubsection{Big Join Results}

We consdier the big join results one at a time, in an arbitrary order.
Let $B$ be a big join result with a size of $M \times N$, where
$(j-1)W/t < MN \leq jW/t$.
We apply a \textbf{result-to-machine} mapping method for $B$ with the number of machines set to $j$.
%
Without loss of generality, let the machines assigned be
$\mathcal{M}_1, ..., \mathcal{M}_j$.
The result of the mapping is that each machine $\mathcal{M}_i$ will be
mapped to a rectangular region in the join result $B$.
Each rectangular region is defined by a quadruple
$\langle l^i_s, h^i_s, l^i_t, h^i_t \rangle$,
where $l^i_s, h^i_s$ are two tuple id's in table $S$, where $l^i_s < h^i_s$,
and $l^i_t, h^i_t$ are two tuple id's in table $T$, where $l^i_t < h^i_t$.
A tuple in table $S$ with id in $[ l^i_s, h^i_s ]$ is assigned to
$\mathcal{M}_i$. Similarly, a tuple in $T$
with id in $[l^i_t, h^i_t]$ is assigned to $\mathcal{M}_i$.
For example, in Figure \ref{fig:join} (b), suppose we divide the join result horizontally into 2 equal sized rectangles. The top rectangle is defined by $\langle 2,3,2,4 \rangle$. Suppose
 this rectangle is assigned to machine $\mathcal{M}_2$.
Then tuples 2 and 3 of $S$,
and tuples 2, 3, and 4 of $T$ will be assigned to $\mathcal{M}_2$.

We divide the $MN$ result tuples among $j$ machines by partitioning the longer side of the rectangle $B$ into $j$ intervals as evenly as possible. Without loss of generality, assume $M \geq N$. Then $M$ is divided into $j$ intervals. Each of the $j$ intervals and the side of size $N$ of region $B$ form a rectangle in $B$.
Hence $B$ is partitioned into $j$ such rectangles.
We call these rectangles the \textbf{mapping rectangles}.
There are two possible cases for the size of $MN$:
\begin{enumerate}
\item
$MN = jW/t$.
In this case, the $j$ mapping rectangles are of the same size of $W/t$. The output of each mapping rectangle are assigned to one of $j$ machines that have not been assigned any big join result so far.
We send the $N$ tuples on the $T$ side of $B$ and tuples along interval $i$, $1 \leq i \leq j$, on the $S$ side of $B$, to $\mathcal{M}_i$.
\item
$MN < jW/t$.
Since we partition the longer side of $B$ (with $M$ tuples) as even as possible, each interval has either $\lceil M/j \rceil$ or $\lfloor M/j \rfloor$ tuples. Thus, the smallest mapping rectangle $R_{min}$ has a size smaller than $W/t$.
For each of the $j-1$ mapping rectangles other than $R_{min}$, the corresponding tuples are processed as in Case (1) above, so that their output are assigned to $j-1$ machines.
For $R_{min}$, it is treated as a small join result, which is to be processed as described in the next subsection.
We call $R_{min}$ a \textbf{residual join result}.
\end{enumerate}

Note that in the above, each machine is assigned at most one mapping rectangle. No rectangles from two or more big join results will be assigned to the same machine.
Also note that the number of machines thus assigned is no more than $t$.

%

\subsubsection{Small Join Results}

After the big join results are assigned to the machines, we deal with the
\textbf{result-to-machine} mapping for the small join results.
The small join results include those residual join results, for the smallest mapping rectangles $C_{min}$ that are generated in Case (2) in the processing of big join results.
We consider small join results for different join keys one by one, each time we assign the next join result to the machine with a smallest assigned workload,
we continue until all results are mapped.
Note that the small join results do not need to be sorted in any order. We shall show that the algorithm terminates with a bound of $2W/t$ for the maximum join result workload on any machine.

\subsubsection{StatJoin Algorithm}

The pseudocode of StatJoin is shown in Figure \ref{mapReduceJoin}.
First the tuples of each table $S$ and $T$ are distributed evenly to
each machine. Each machine sorts its portions of data on the join key in both
$S$ and $T$ by adopting a parallel sort mechanism such as
Terasort or SMMS.
In Step 2, after the sorting, statistics are collected at each
machine and sent to a file folder $F2$.
Step 3 determines the result-to-machine mapping based on
the statistics. Step 4 applies the mapping to send tuples
of $S$ and $T$ to the mapped machines.
In step 5, each machine generates join results from the received
tuples.
When implementing StatJoin under the Hadoop MapReduce framework, Steps 1 and 2 can be implemented as 2 rounds.
Steps 3,4, and 5 can be implemented as one single MapReduce round,
where Step 3 is incorporated in 
the map setup function.
Step 4 is a map phase, and Step 5 is a reduce phase.

\subsubsection{Analysis}

Next we examine some useful properties of StatJoin and analyze the algorithm by means of $(\alpha,k)$-minimality.

\begin{lemma}
Let $W$ be the total size of all join results.
Given a big join result $B$ with a size of $M \times N$.
If $max(M,N) \geq t$, then excluding any residual join result from $B$,
the maximum number of tuples from $B$, $w_B$, assigned to any machine by
the result-to-machine mapping is less than $2W/t$.
\label{lem:BigJoinResult}
\end{lemma}

\begin{proof}:
Without loss of generality, assume $M > N$.
Let $(j-1)W/t < MN \leq jW/t$.
Hence $max(M,N) = M \geq t \geq j$. $M/j \geq 1$.
The maximum size $w_B$ is given by $\lceil M/j \rceil \times N$.
$\lceil M/j \rceil \times N < (M/j+1)N = MN/j + N$.
Since $M/j \geq 1$, we have $N \leq MN/j$.
Hence the maximum size $w_B$ is less than $2MN/j$.
Since $MN \leq jW/t$, the maximum size is
less than $2W/t$. $\qed$
\end{proof}




\begin{figure}[htbp]
\hrulefill
\\
{\texttt{\bf StatJoin}}

\vspace*{-2mm}
\hrulefill

\textbf{Rounds 1 and 2:}

\begin{itemize}[leftmargin=.1in]
\item[]
\textbf{Step 1 :}  Same as Steps 1 to 3 of Terasort or SMMS Sorting.
\end{itemize}


\begin{itemize}[leftmargin=.1in]
\item[]
{\textbf{Step 2 :}}
$\mathcal{M}_i$ first sorts the data (with Terasort or SMMS). 
Then $\mathcal{M}_i$ generates $(k, k_{id},\gamma)$ to folder $F_1$,
where $k_{id}$ is the id of a tuple with key $k$
in table $\gamma$.
$\mathcal{M}_i$ generates also the statistics of $(k, total_k,\gamma)$
to folder F2, where $\gamma$ is the table id,
$total_k$ is the number of tuples with key $k$ in table $\gamma$.
\end{itemize}

\medskip

\textbf{Round 3}

\begin{itemize}[leftmargin=.1in]
\item[]
{\textbf{Step 3 :}} 
With the statistics from Step 2, compute the result-to-machine mapping.
Generate mappings of the form
$(k, \langle l_s, h_s, l_t, h_t, i \rangle)$, i.e. the range of tuple ids is from
$l_s$ to $h_s$ in $S$ and from $l_t$ to $h_t$ in $T$, and such tuples
are mapped to $\mathcal{M}_i$.
\end{itemize}


\begin{itemize}[leftmargin=.1in]
\item[]
{\textbf{Step 4 :}}
$\mathcal{M}_i$ follows the result-to-machine mapping to
assign tuples from the local storage (allocated portion of data) to all machines for the next step.
\end{itemize}


\begin{itemize}[leftmargin=.1in]
\item[]
{\textbf{Step 5 :}}
$\mathcal{M}_i$ generates join results from the received tuples by the cross product operation.
\end{itemize}

\hrulefill
\caption{StatJoin - Deterministic algorithm for Skew Join}
\label{mapReduceJoin}
\end{figure}

\begin{theorem}
Let the total join results size be $W$. With StatJoin, the total size of the join results generated by any machine
is at most $2W/t$.
\label{lem:statjoin}
\end{theorem}

\begin{proof}:
In the following, we refer to the join result size as
\emph{work}.
Note that the biggest work load of a small join result is $W/t$.
We prove by contradiction.
Suppose a machine $\mathcal{M}$ has $W' > 2W/t$ work.
From Lemma \ref{lem:BigJoinResult}, after the
mapping for big join results, each machine is assigned no more than $2W/t$ work, hence the last join result assigned to $\mathcal{M}$ is a small join result.
When the last small join result $B_{last}$ of size $w$  is assigned to $\mathcal{M}$, all other
machines must have at least $W' - w$ work each,
since otherwise, $B_{last}$ should be assigned to another machine and not $\mathcal{M}$.
Hence, the smallest total work assigned to all other
 machines is $(W'-w)(t-1) = W - W/t + t-1$ (when $W' = 2W/t +1$, and $w = W/t$).
The work $W' \geq 2W/t+1$ is
assigned to $\mathcal{M}$.
Thus, the total work assigned to all machines is greater than $W$,
we arrive at a contradiction since the total work is only $W$.
$\qed$
\end{proof}

To our knowledge Lemma \ref{lem:statjoin} gives the best workload distribution guarantee among all known algorithms for skew join from previous works.

\begin{theorem}
Given a join skew factor of $\sigma$,
StatJoin is $(3,(2+t/\sigma))$-minimal if for each big join result $B$,
if the size of $B$ is $M \times N$, then
$\max(M,N) \geq t$.
\label{thm:StatJoin}
\end{theorem}

As in
Theorem \ref{thm:randjoin},
$t/\sigma$ covers the input size of a machine.

\section{Experimental Results}  \label{sec:exp}

We report the results of our experiments to evaluate our proposed algorithms with an objective to verify our analysis based on $(\alpha,k)$-minimality.
Our experiments for the parallel algorithms have been conducted on a 16 machine cluster with a master machine and 15 slave machines.
The master is a
Dell R720 Server with Dual 6-core Xeon E2620 2.0GHz, 192GB RAM and 4x 3TB SAS Hard Disk.
Each slave machine is a Dell R620 Server - Dual 6-core Xeon E2620 2.0GHz, with 48GB RAM and
2x 300GB SAS Hard Disk.
All machines are connected by a 1GB-ethernet switch.
We have installed Hadoop (version 1.2.1) on the cluster for MapReduce algorithms. There are 6x2x15 = 180 cores in the slaves, we can activate up to 180 workers in parallel for Hadoop mappers or reducers.
For sequential algorithms we have run our jobs on a PC with Intel(R) Core(TM)i7-4770 3.4GHz, 4GB RAM and a 500GB hard disk.

We have implemented the sorting algorithms (SMMS and Terasort) based on MPI, and the join algorithms RandJoin and StatJoin based on Hadoop MapReduce. The maximum number of reducers we use is 180 and we notice that Hadoop assigns reducers evenly to the 15 machines in the cluster so that with 180 reducers, each core in the cluster will be assigned one reducer. We have set the
DFS dfs.replication factor to 3, so that for each data file, 3 duplicated copies will be maintained by HDFS. It also means that whenever we write to a file, the system writes to 3 different copies at the same time. For failure resilience, HDFS will keep the 3 copies at different slave machines.
As we shall see, this has a certain amount of impact on the overall performance. The fs.block.size is set to 64MB.
Other Hadoop parameters are set to the default values.

The computer cluster consists of 15 worker machines each with 8 cores that share 2 hard disks. For massive data the data transfer to and from the hard disks is a major cost, and though we can utilize a maximum of 180 cores, the number of hard disks we can use is only 34 (including the 4 hard disks on the master node). This means that the maximum speedup effect cannot scale up to 180 as the number of cores, but only to some factor between 34 (or less) and 180, depending on the CPU workload versus the I/O workload for the algorithm. Due to this mismatch of our computer cluster with a typical cluster model,
we shall call the parallel computational units \emph{processes} instead of \emph{machines} in our experiments.

We evaluate our algorithms by two measurements: the workload distribution and the runtime.
For sorting, the workload is measured by the input size. For join, we measure the workload by means of the join result size. The sizes are given in the number of tuples unless otherwise specified.
We examine the \textbf{workload imbalance} which is given by the ratio of the maximum workload on a machine versus the even workload.
For the \textbf{runtime}, it is given by the longest runtime taken by any process, and in all experiments it is the runtime of the process given the maximum workload.

\subsection{Results for Sorting}

We evaluate the sorting algorithms of SMMS and Terasort on a real dataset LIDAR and also on a synthetic dataset.
For SMMS, we set the value of $r$ to 1 so that each process samples $t$ objects.
For Terasort, we have implemented two versions: one with Algorithm $\mathbb{S}$, and one without Algorithm $\mathbb{S}$ (as in \cite{Tao13sigmod,Malley08Yahoo}).
Our results show that the two alternatives give very similar
partitioning of the giving dataset $S$, and hence very similar
workload distributions and also overall runtime.
In our report we shall focus on Terasort with sampling algorithm $\mathbb{S}$.
We vary the number of processes from 15 to 180, and measure both the workload distribution and the runtime performance.



{\bf Real Data}:
We use the real dataset LIDAR\footnote{Downloadable from
 http://www.ncfloodmaps.com} for experiments on sorting.
 This dataset has also been used for the sorting experiments in
 \cite{Tao13sigmod}.
 LIDAR contains 8.27 billion records, each of which is a 3D point representing a location in North Carolina. We sort the records by the first dimension. The dataset size is 123GB.

\begin{figure}[!t]
\begin{center}
\hspace*{-2mm}
\includegraphics[width = 1.65in,height=1.25in]{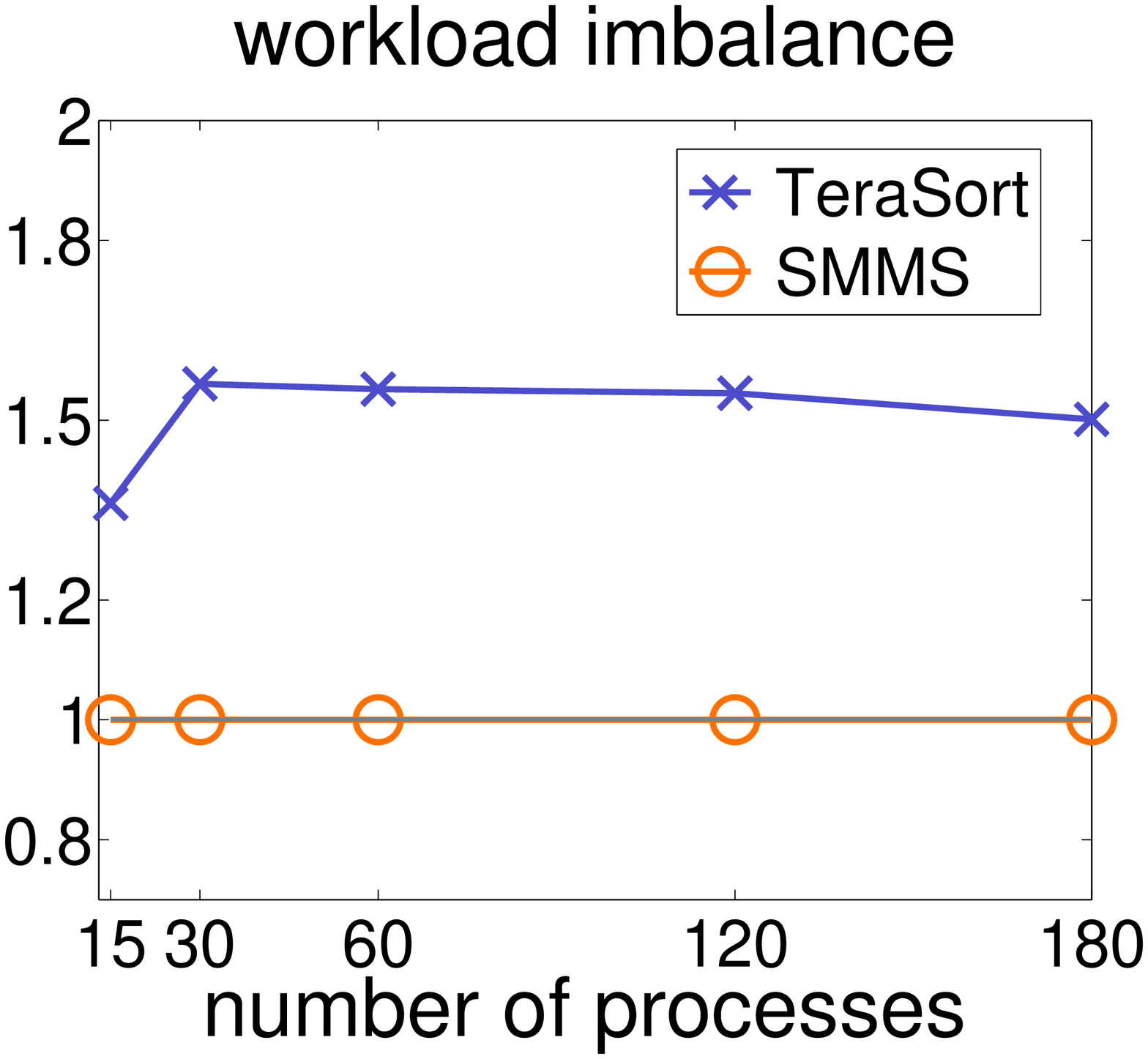} 
\includegraphics[width = 1.65in,height=1.25in]{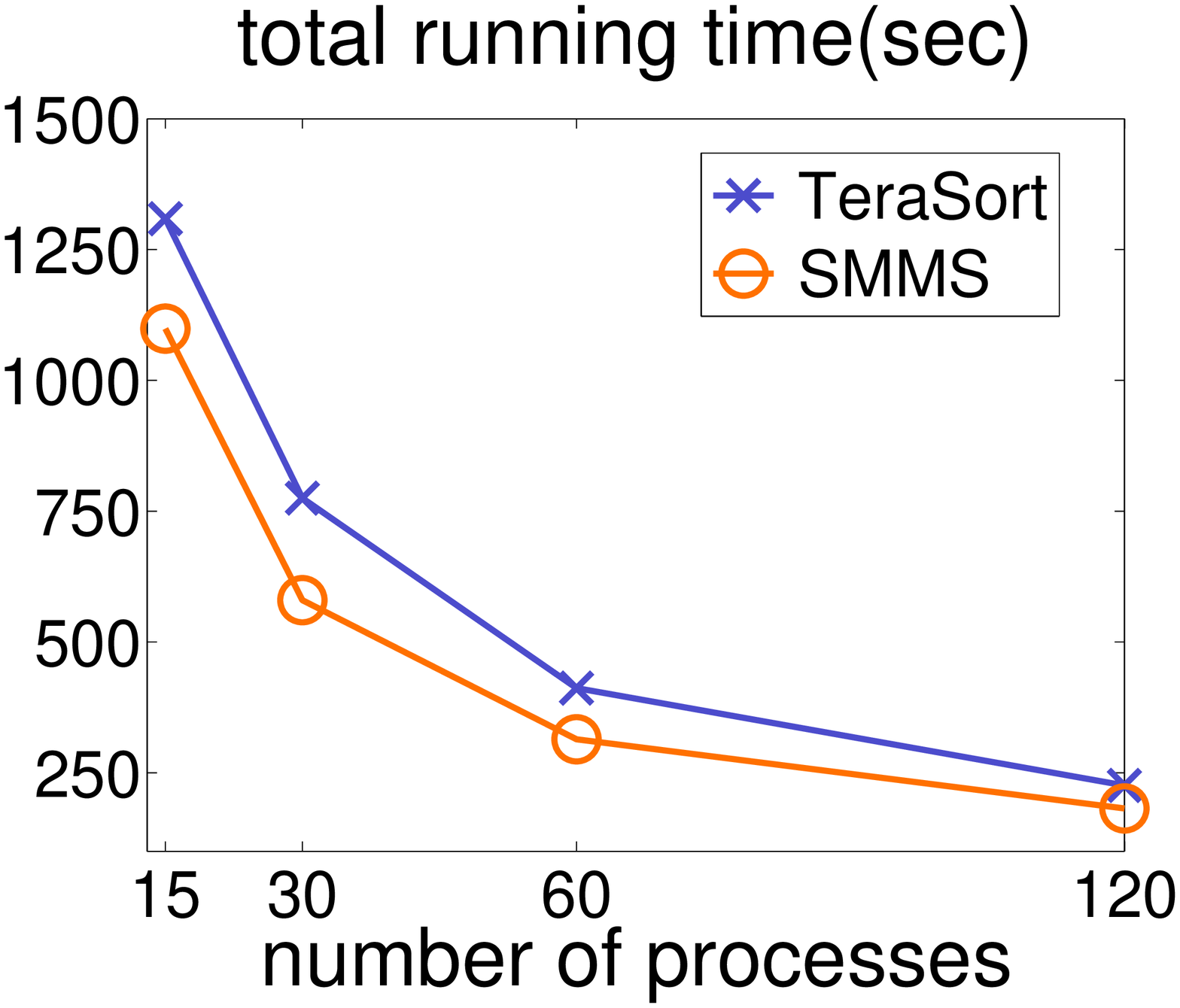} 
\\
\hspace*{-5mm} (a) workload (input size)
\hspace*{15mm}
(b) run time \ \ \
\hspace*{10mm}\\
\caption{Sorting for real dataset LIDAR, workload imbalance = maximum workload / optimal workload }
\label{exp:sortReal}
\end{center}
\end{figure}

%


\begin{figure}[!t]
\begin{center}
\hspace*{-2mm}
\includegraphics[width = 1.65in,height=1.25in]{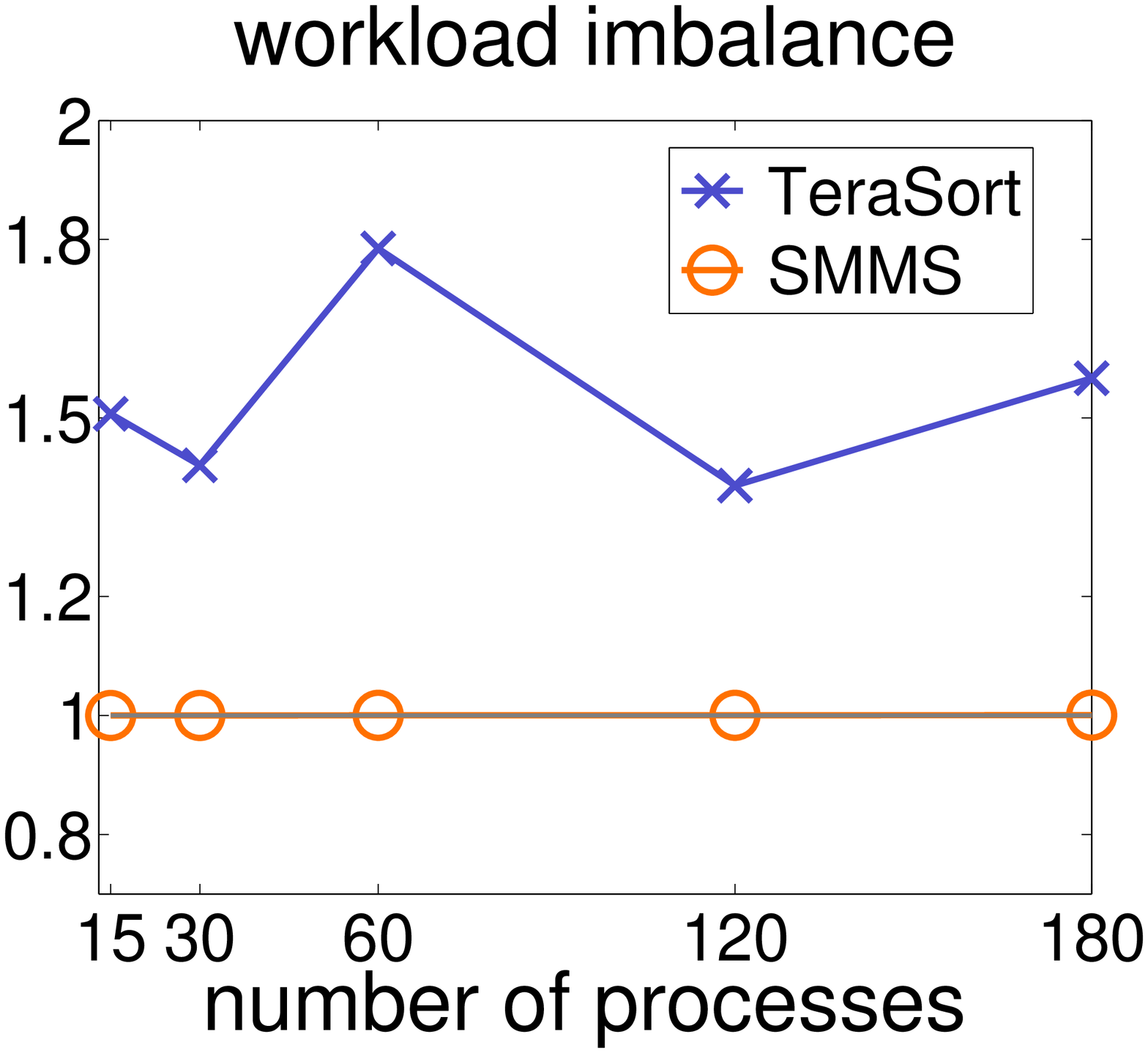}
\includegraphics[width = 1.65in,height=1.25in]{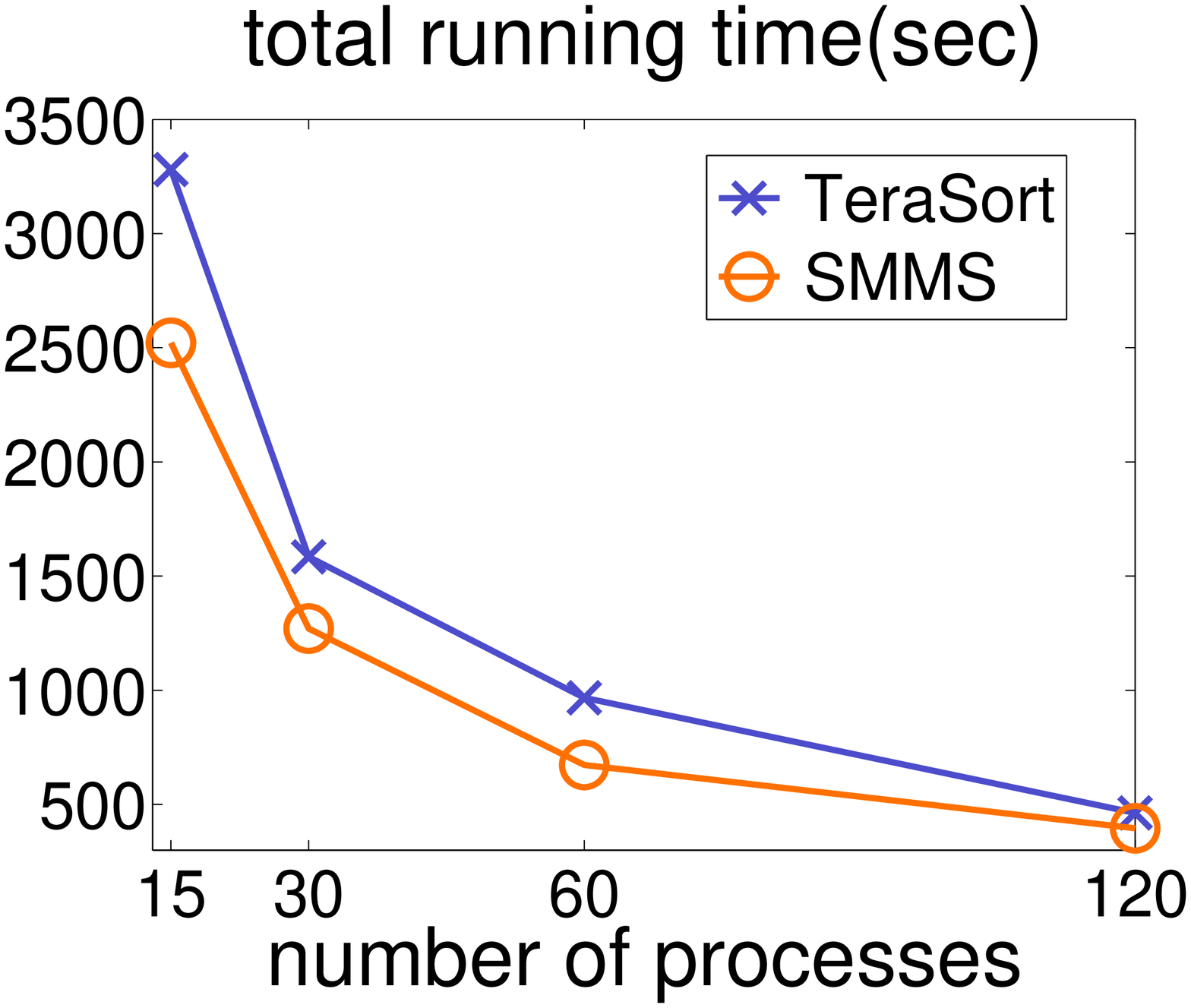}
\\
\hspace*{-5mm} (a) workload (input size)
\hspace*{15mm}
(b) run time (sec) \ \ \
\hspace*{10mm}\\
\vspace*{-1mm}
\caption{Comparing SMMS and Terasort for Random Dataset with 18 billion objects (199.3 GB)}
\label{exp:sortRand}
\end{center}
\end{figure}

\begin{figure}[!t]
\begin{center}
\hspace*{-2mm}
\includegraphics[width = 1.65in,height=1.25in]{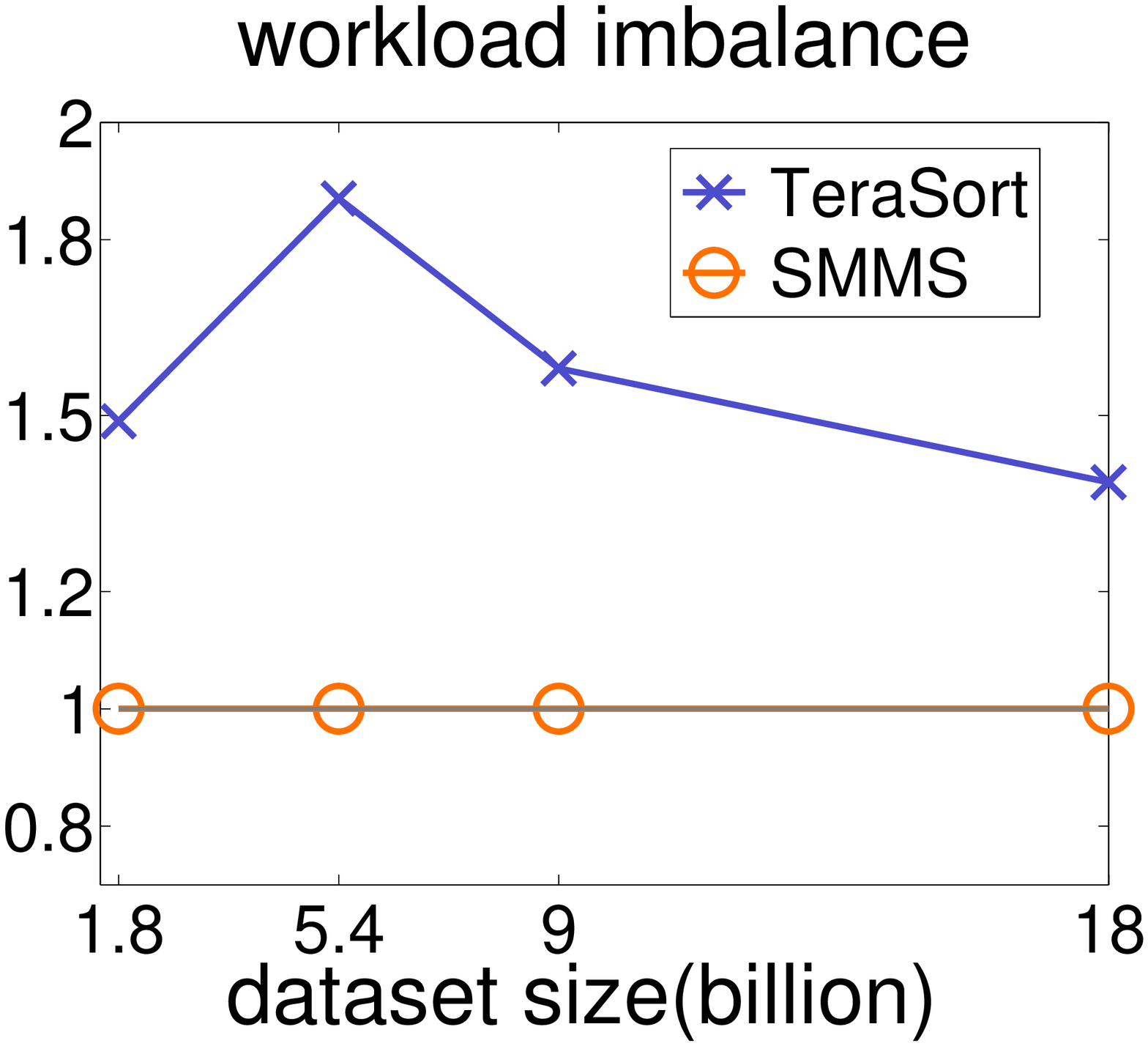}
\includegraphics[width = 1.65in,height=1.25in]{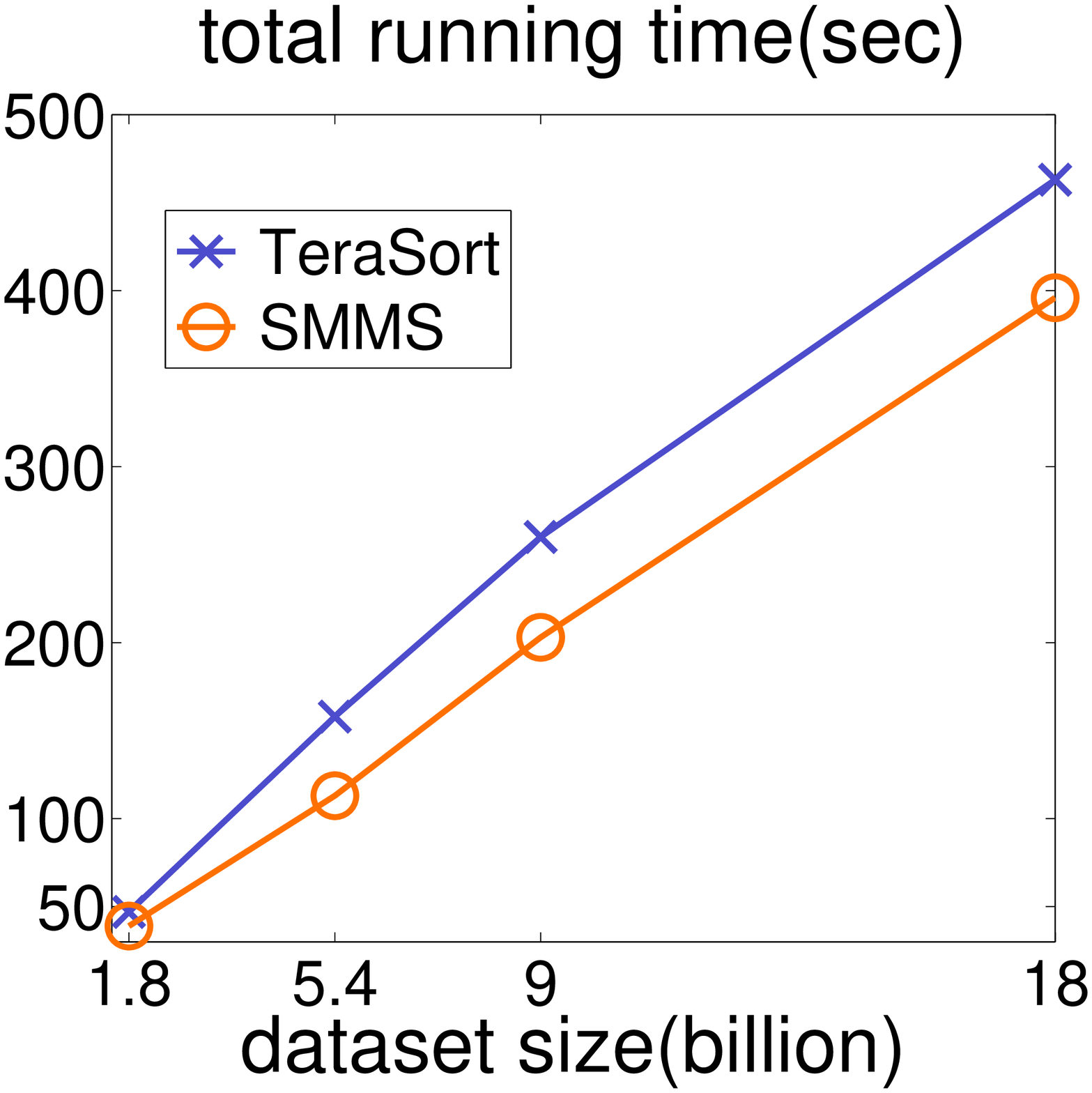}
\\
\hspace*{2mm} (a) workload (input size)
\hspace*{15mm}
(b) run time (sec) \ \
\hspace*{10mm} \\
\caption{Sorting results for Random Datasets of different sizes with 120 processes}
\label{exp:sortSizes}
\end{center}
\end{figure}

\textbf{Synthetic Data} :
We have generated 4 sets of random data, with
1.8 billion objects, 5.4 billion objects,
9 billion objects and 18 billion objects.
The sizes of these datasets
are 19.9 GB, 59.9 GB, 99.8 GB and
199.3 GB, respectively.
The key of each data object in a dataset is a randomly generated number in the range of $[1,12 \times 10^6]$.
We generate unique objects in each machine.

\subsubsection{Workload Imbalance}
The results of workload imbalance are shown in Figures
\ref{exp:sortReal}(a), \ref{exp:sortRand}(a), and \ref{exp:sortSizes}(a).
In all cases, SMMS distributes the workload very evenly and the imbalance is close to the optimal value of 1.
TeraSort has comparably much larger workload imbalance, in most cases the maximum workload of a process is above 1.5 of the optimal load. The imbalance affects the performance in runtime.
Another negative effect of the imbalance is the need of larger storage or main memory for TeraSort for supporting the larger data size on a cluster node.
The results show the superiority of the bucket boundary selection of SMMS as compared to TeraSort.

%
%

\subsubsection{Runtime Comparison}

The runtime results are shown in Figures \ref{exp:sortReal}(b),
\ref{exp:sortRand}(b), and \ref{exp:sortSizes}(a).
It can be seen that SMMS achieves almost linear speedup, going from 15 to 30 processes almost halved the runtime and similarly going from 30 to 60, and 60 to 120 processes.
This result is a consequence of the highly even workload distribution.

\begin{table}
\begin{center}
\begin{footnotesize}
\begin{tabular}{|l|r|r|r|r|r|}
  \hline
  Dataset & S1.8b & S5.4b & S9b & S18b & LIDAR \\ \hline \hline
  $\mathcal{A}_{seq}$ & 1540s & 4718s & 7914s & 15911s & 8405s \\ \hline
  {\tiny SMMS}(15)& 237s & 715s & 1254s & 2522s & 1099s \\
  \hline
  {\tiny SMMS}(30)& 123s & 369s & 648s & 1270s & 577s\\
  \hline
  {\tiny SMMS}(60)& 66s & 198s & 343s  & 673s & 314s \\
  \hline
  {\tiny SMMS}(120) & 39s & 113s & 203s & 396s & 182s \\
  \hline
\end{tabular}
\end{footnotesize}
\end{center}
\caption{Total runtime for sequential sorting and SMMS($t$), where $t$ is the number of processes}
\label{tab:runtime}
\end{table}

The runtime of sequential sorting $\mathcal{A}_{seq}$
on our PC and the runtime of SMMS are listed in Table \ref{tab:runtime}. In this table
S$x$b stands for the synthetic dataset with $x$ billion tuples.
From these results, the time for running SMMS on 15 processes is about
1/6 to 1/8 of the sequential time. The speedup is nearly linear considering that the PC has a much faster CPU compared to the cluster machines, and also local disk I/O is about 10 times faster compared to network transmission time with our machines. 

\subsection{Results on Skew Join}

For the Skew Join experiments the dataset consists of two input tables $S$ and $T$. 
We adopt two different methods to form a dataset with skew join keys.
The first method is to generate tables with attributes drawn from the Zipf distribution and maintaining the same distribution for both tables so that each key has the same frequency in both of the input tables. We shall vary the Zipf skew parameter $\theta$ between 0 (skew) and 1 (uniform), i.e., $Z(r) \propto 1/r^{(1-\theta)}$, where $r$ is a frequency rank, $Z(r)$ is the frequency of the item with rank $r$.

The second kind of skew data is generated as described in
\cite{DeWitt92vldb}. For a table with $n$ tuples, the join key
has a domain of  $[n,2n)$.
The special join key 
$n$ appears in a fixed number of tuples, while the remaining tuples are randomly assigned a
join key from $[n,2n)$.
The output tuple size is 95 bytes.
The skew key $k_0 = n$
is generated in both tables $S$ and $T$, and it occurs $M$ times in
$S$ and $N$ times in $T$. By adjusting $M$ and $N$ we can control the expected output join sizes.
This kind of test data is called ``scalar skew''
in \cite{Walton91vldb} and is also used in
the study in \cite{Omiecinski91vldb}.




\begin{figure}[!t]
\begin{center}
\begin{small}
\hspace{-2mm}
\includegraphics[width = 1.65in]{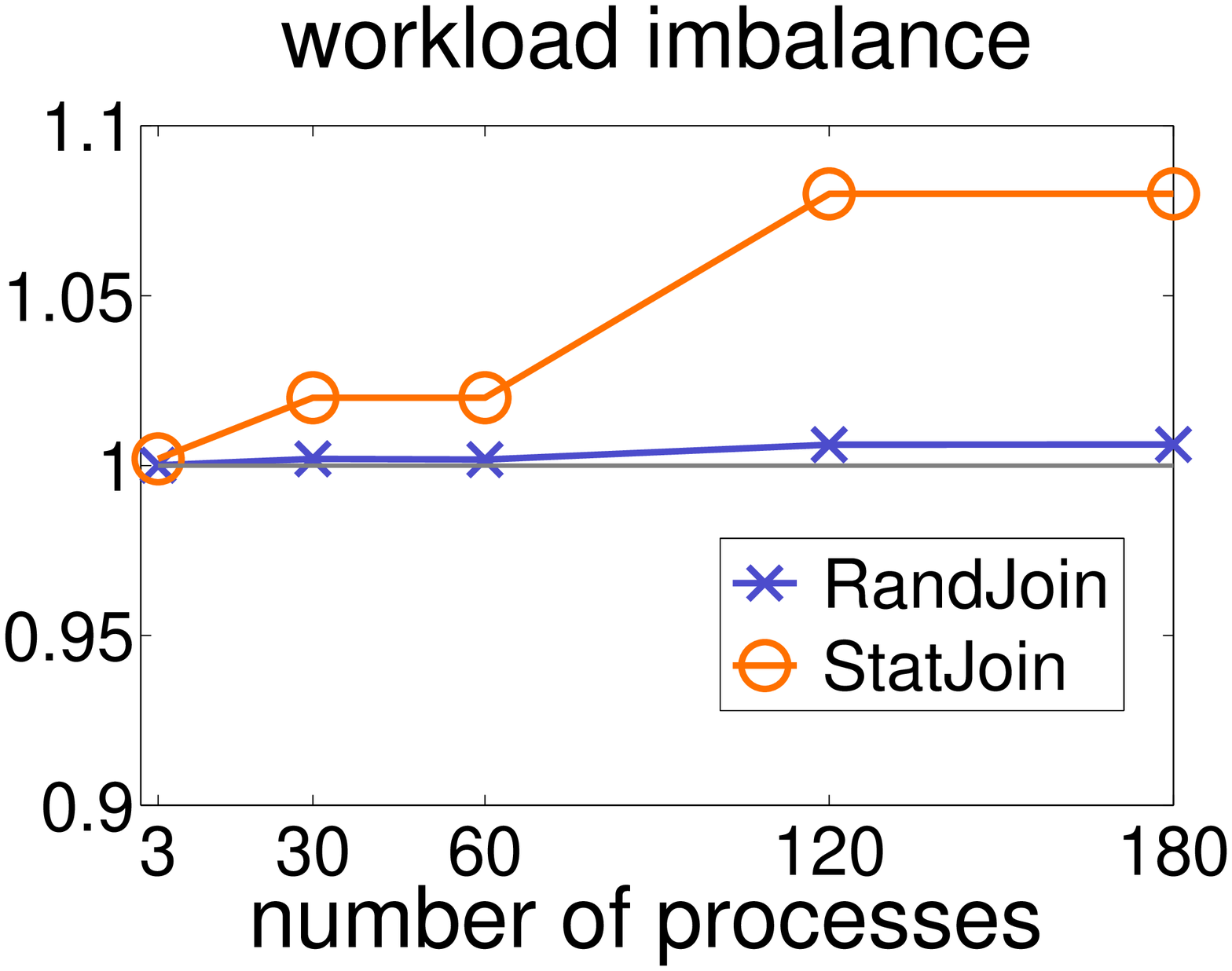}
\includegraphics[width = 1.65in]{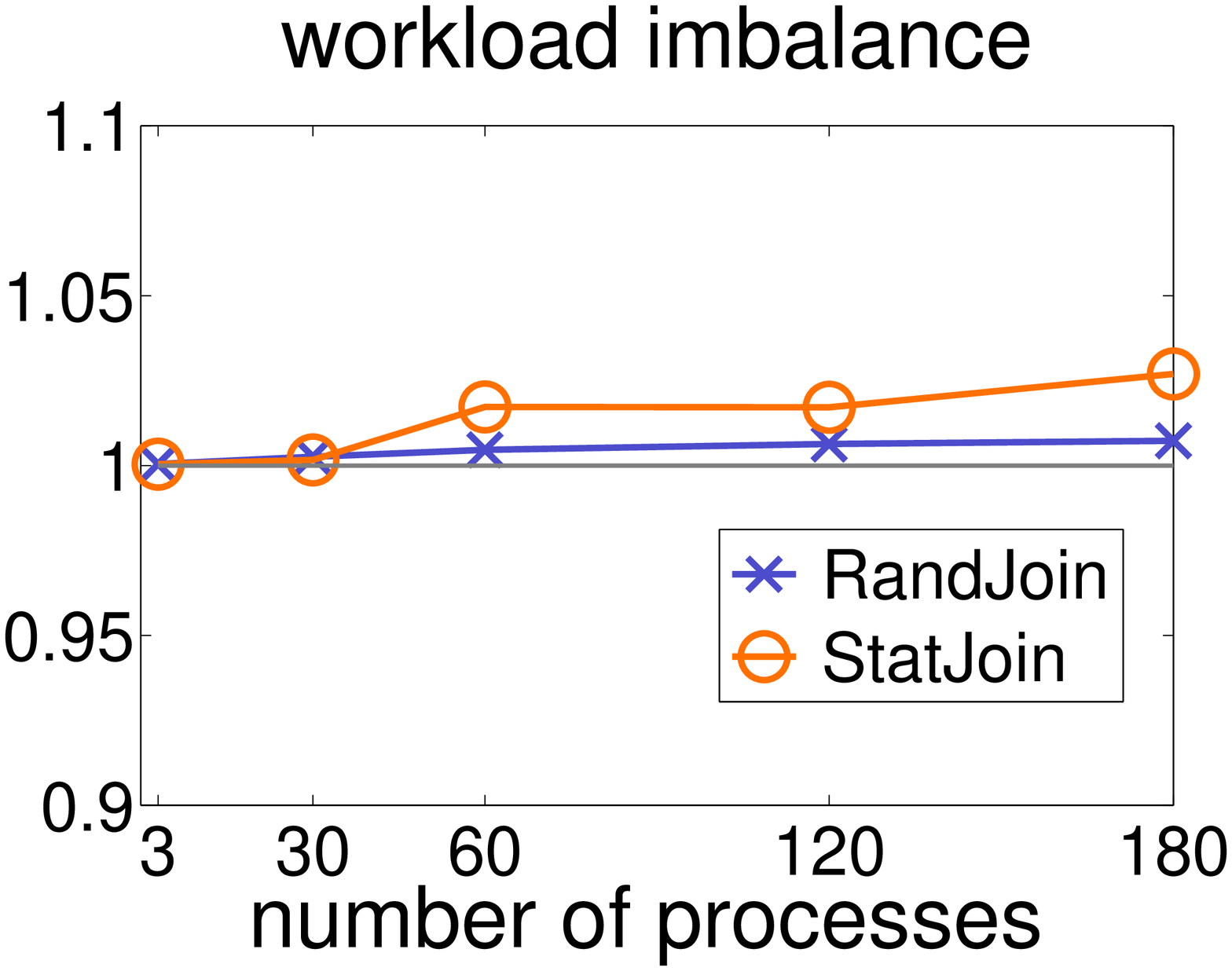}
\\
\hspace*{1mm} (a) $\theta$ = 1, $|S|=|T|$=5M
\hspace*{10mm}
(b) $\theta = 0.7$, $|S|=|T|$=5M \\
output size = 125GB
\hspace*{15mm}
output size = 147GB\\
($25.0$$\times 10^9$ tuples, $\sigma$=2500)
\hspace*{3mm}
(29.4$\times 10^9$ tuples, $\sigma$=2940)\\
\vspace*{3mm}
\hspace*{-2mm}
\includegraphics[width = 1.65in]{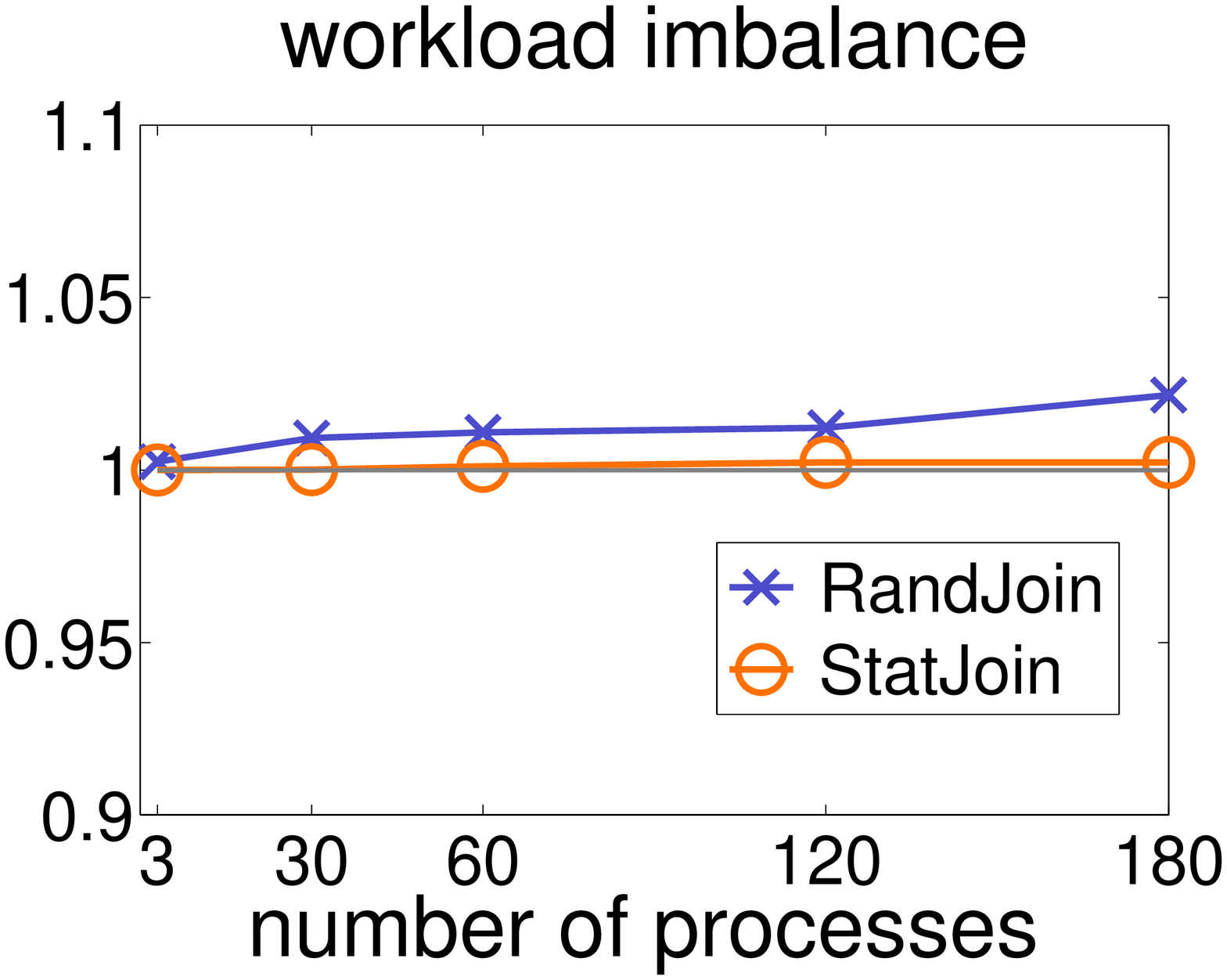}
\includegraphics[width = 1.65in]{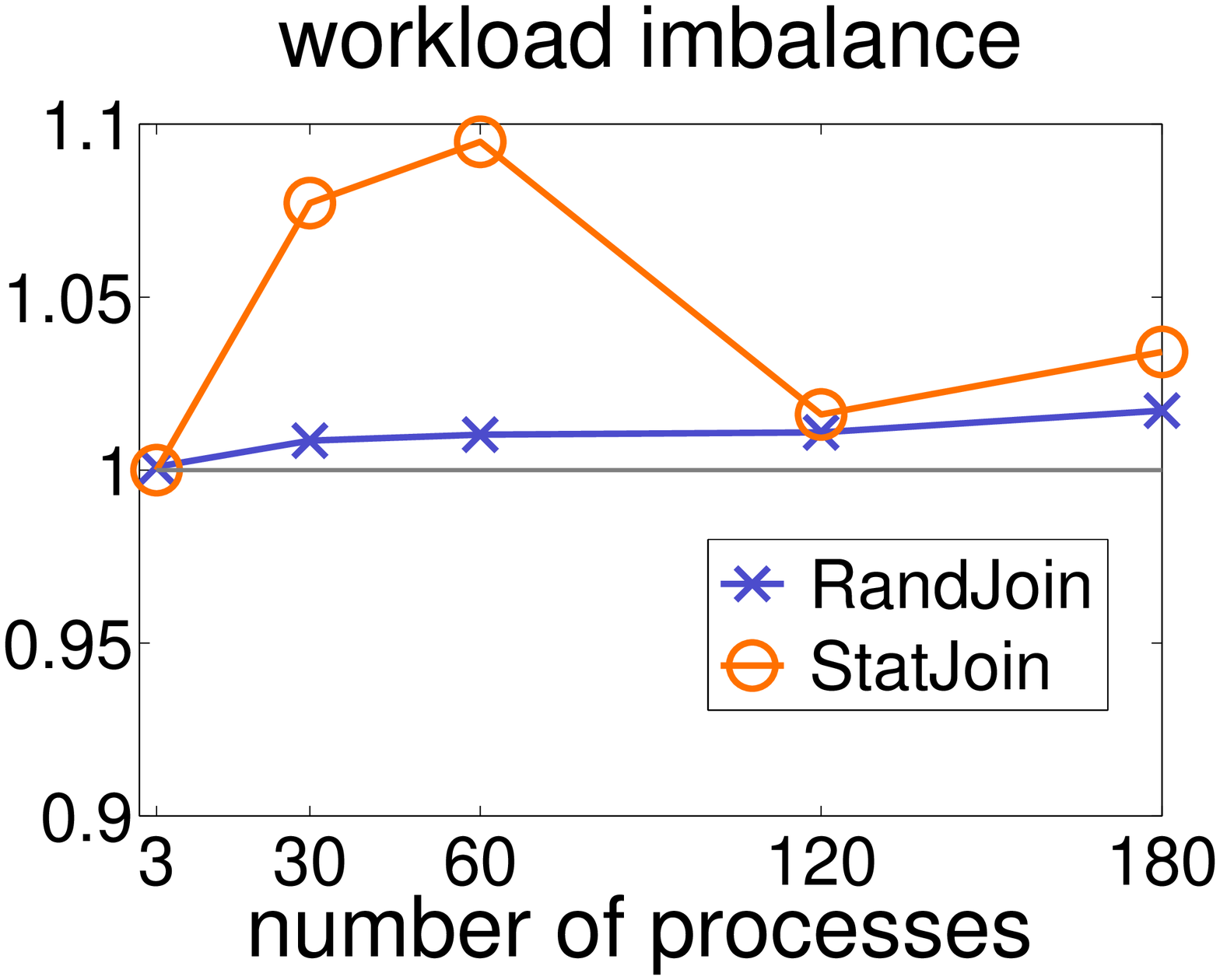}
\\
(a) $\theta$ = 0.3, $|S|=|T|$=1.5M
\hspace*{3mm}
(b) $\theta = 0$, $|S|=|T|$=1.5M
\hspace*{15mm}\\
output size = 59GB
\hspace*{15mm}
output size = 330GB\\
\ (11.8$\times 10^9$ tuples, $\sigma$=3900)
\hspace*{3mm}
(66.0$\times 10^9$ tuples, $\sigma$ = 22000)\\
\caption{Workload distribution of RandJoin and StatJoin for Zipf distributions: ($\theta$ = 1 : uniform key distribution). Workload corresponds to join result size.} \label{exp:JoinZipfwork}
\end{small}
\end{center}
\end{figure}

\begin{figure}[!t]
\begin{center}
\hspace*{-2mm}
\includegraphics[width = 1.65in]{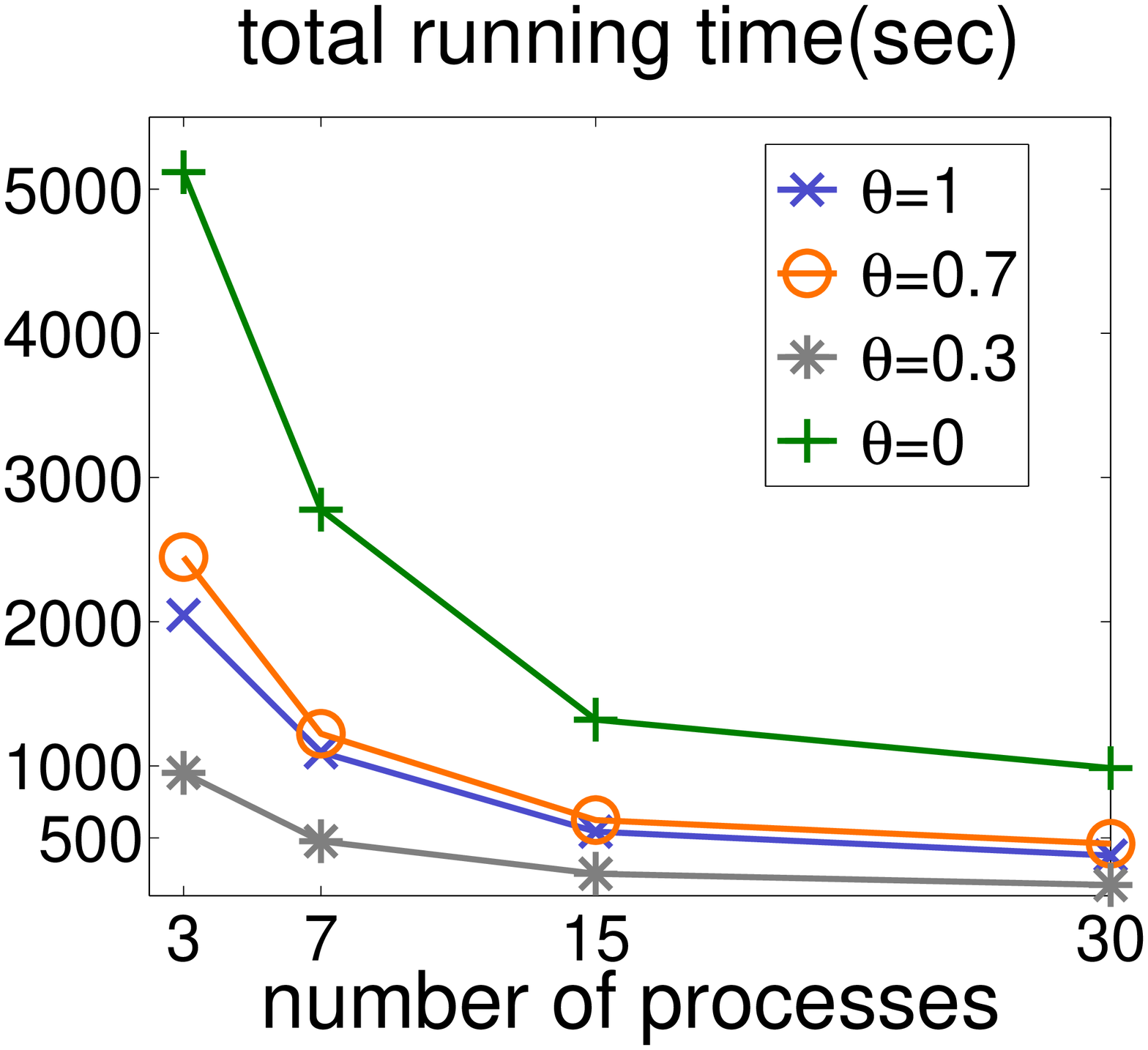}
\includegraphics[width = 1.65in]{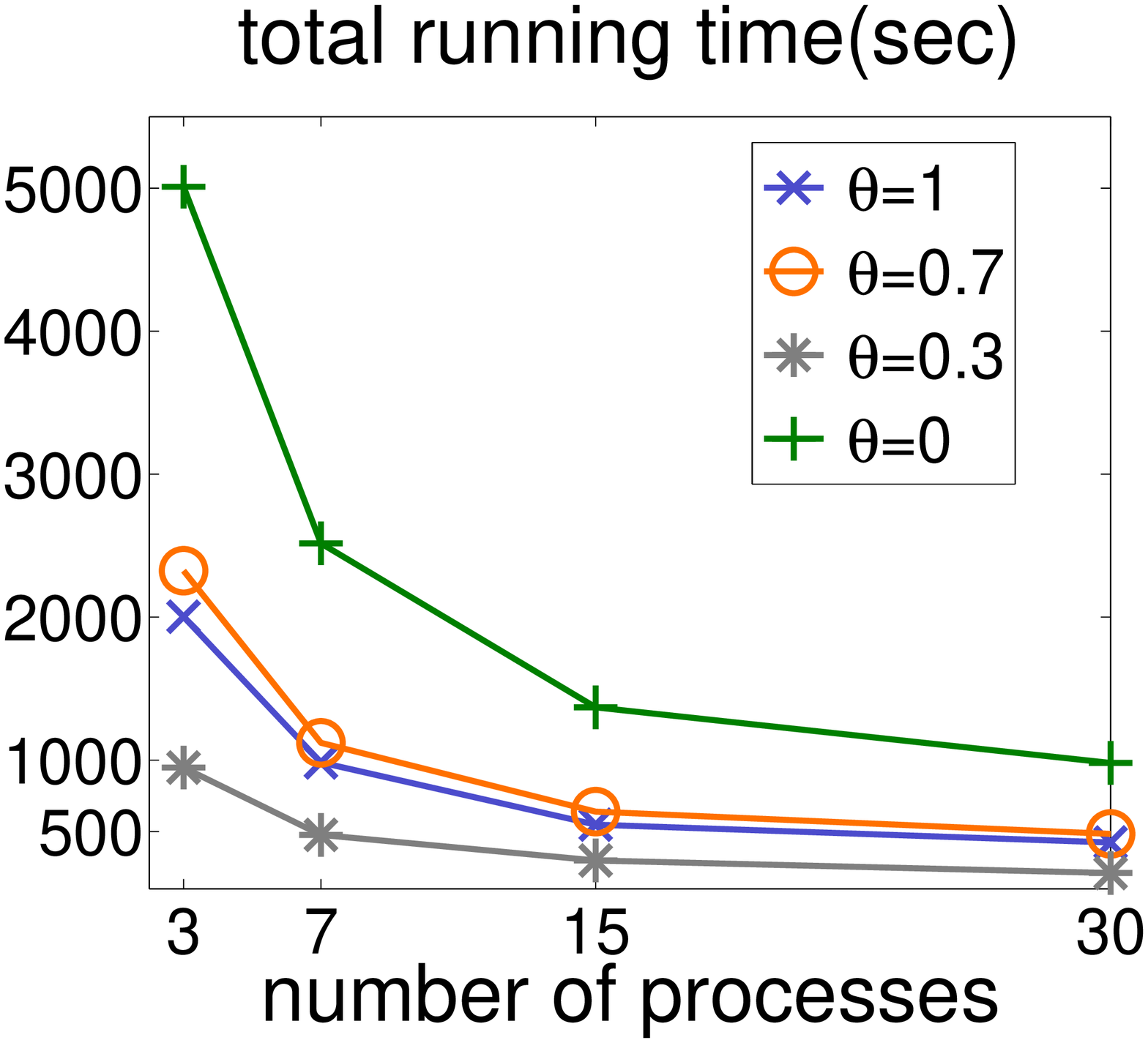}
\\
\hspace*{1mm} (a) RandJoin
\hspace*{30mm}
(b) StatJoin
\hspace*{15mm}\\
\caption{Running time for Zipf skew datasets (in sec)} \label{exp:JoinZipftime}
\end{center}
\end{figure}

\textbf{Zipf distributed dataset}:
We aim to compare the effect of skewness on similar join output size. However, Zipf distributions would vary the output size for the same input size. Therefore we vary the input table sizes as follows. For $\theta$ values below 0.5 we use two tables with 5 million tuples each.
For $\theta$ values above 0.5, we use two tables with 1.5
million tuples each. Following the design of \cite{Okcan11sigmod} for skew key distribution,
each tuple contains a 4 byte join key with a domain of
$[1000, 1999]$.




\begin{figure}[!t]
\begin{center}
\hspace{-2mm}
\includegraphics[width = 1.65in,height=1.25in,height=1.1in]{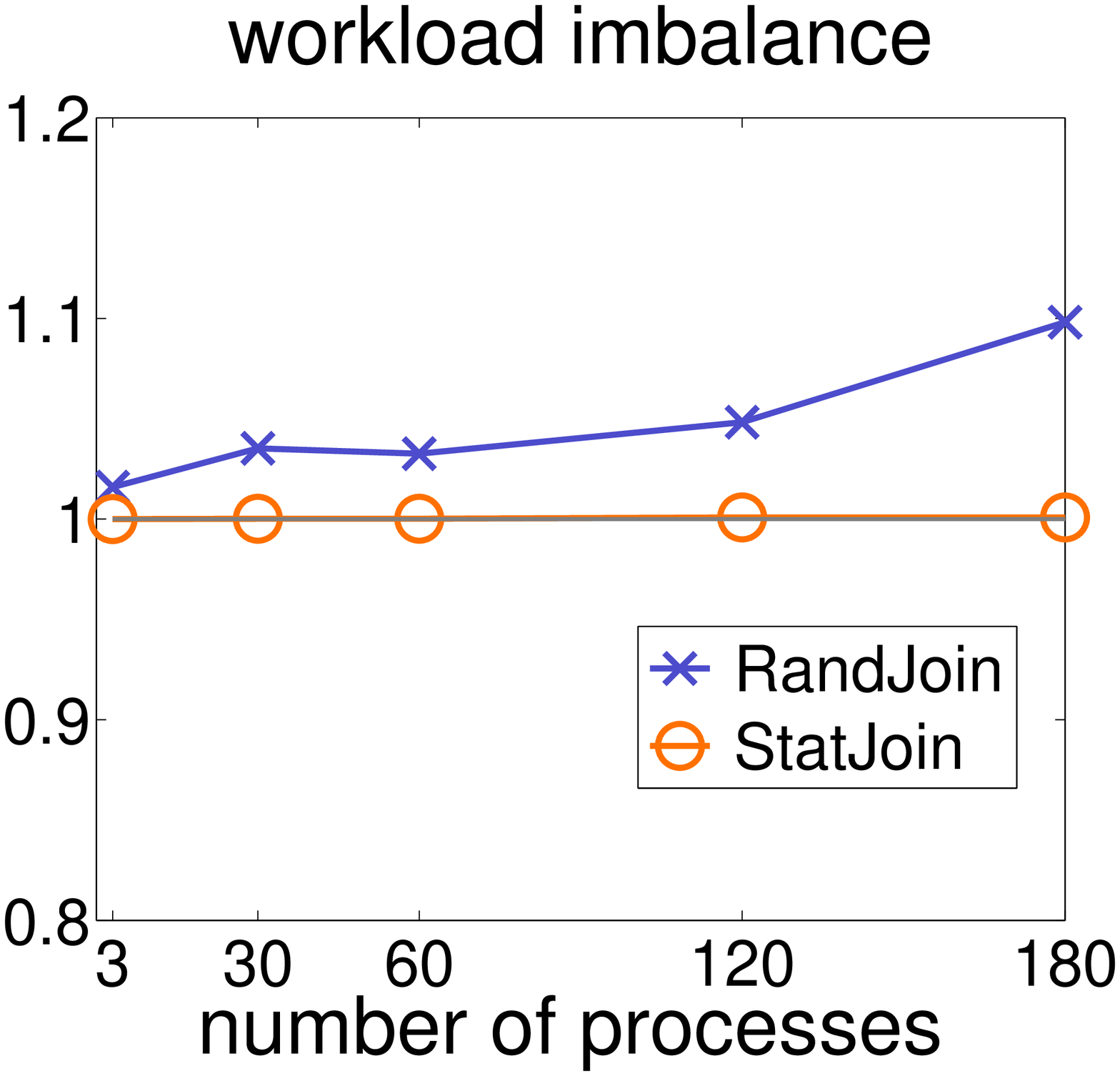} 
\includegraphics[width = 1.65in,height=1.25in,height=1.1in]{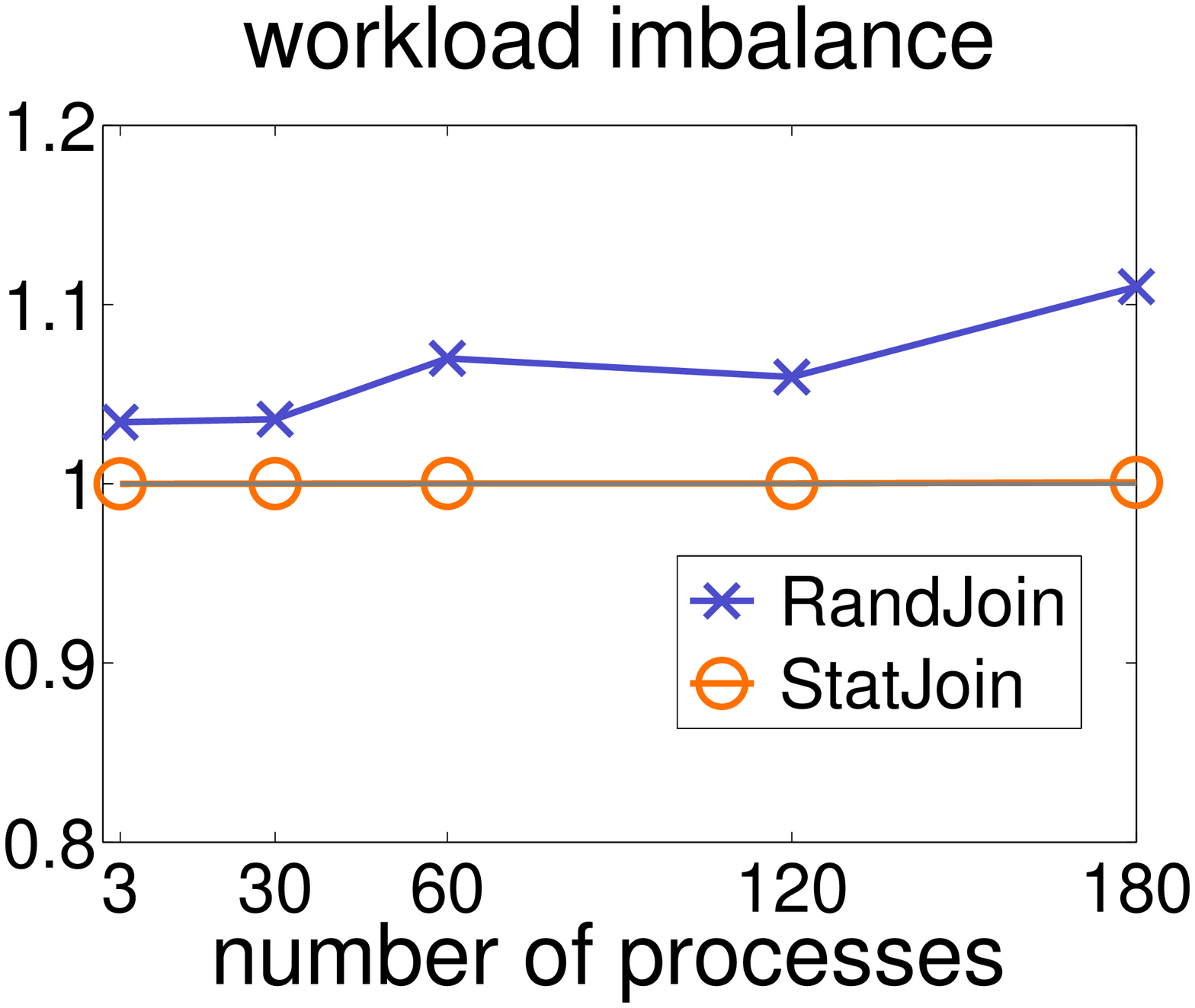} 
\\
(a) $M = 10^5$, $N=2 \times 10^4$
\hspace*{3mm}
(b) $M = 2 \times 10^5$ $N = 10^4$\\
\caption{Workload distribution for scalar skew data. Workload corresponds to join result size.
 } \label{exp:Joinskewload}
\end{center}
\end{figure}

\begin{figure}[!t]
\begin{center}
\hspace*{-2mm}
\includegraphics[width = 1.65in,height=1.3in]{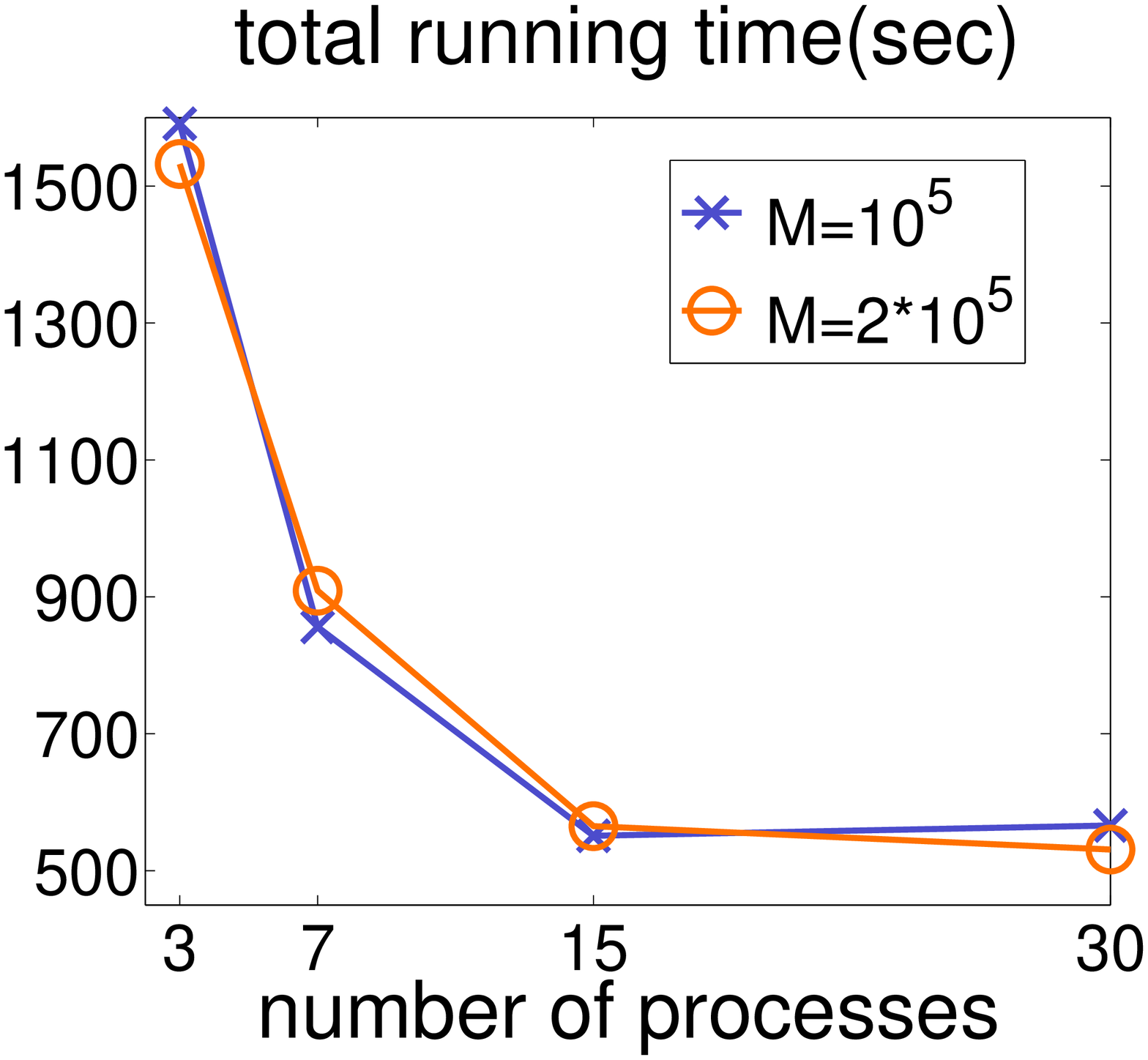}
\includegraphics[width = 1.65in,height=1.3in]{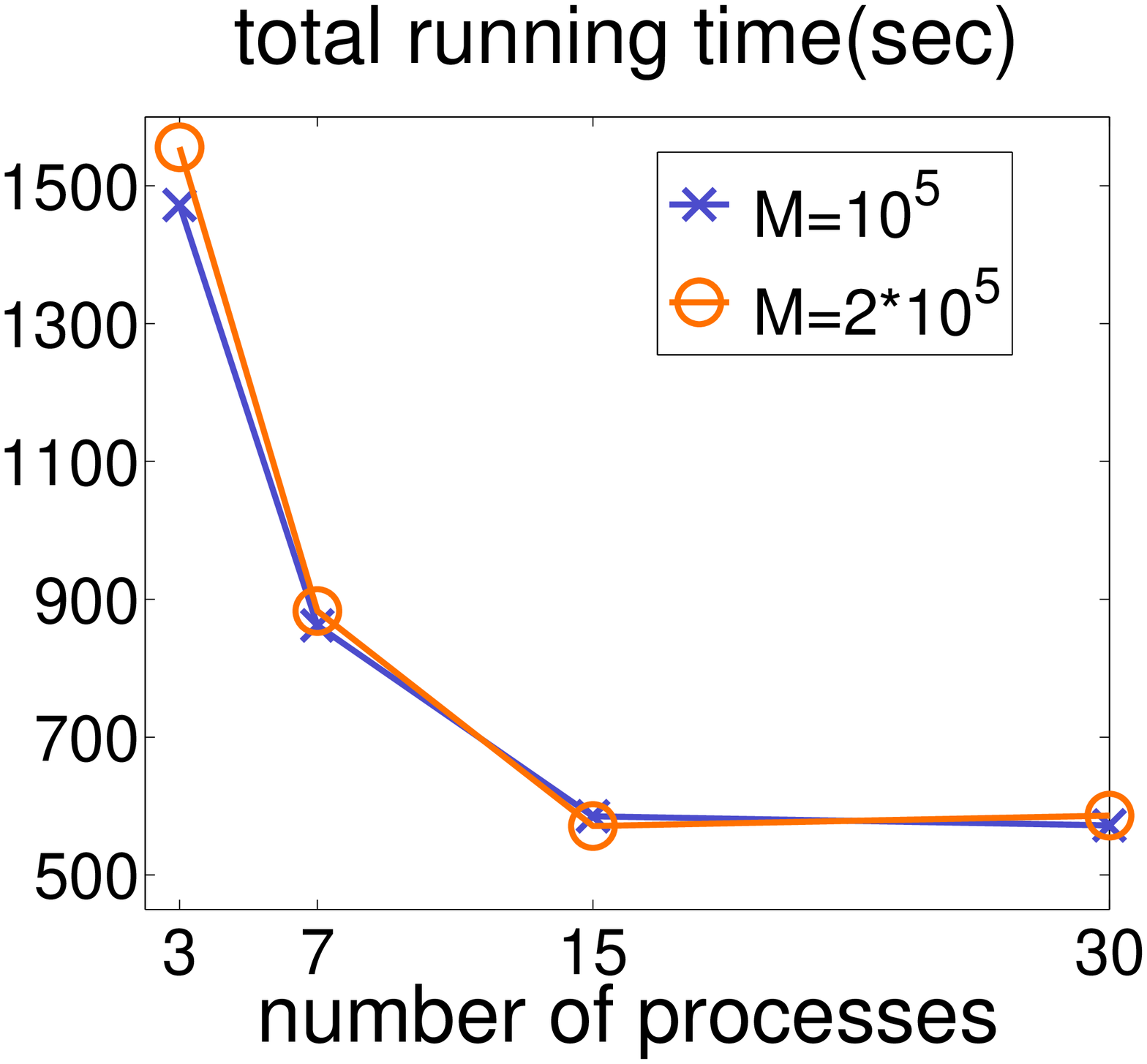}
\\
\hspace*{1mm} (a) RandJoin
\hspace*{30mm}
(b) StatJoin
\hspace*{15mm}\\
\caption{Running time for scalar skew datasets (in sec)} \label{exp:Joinskewtime}
\end{center}
\end{figure}

\textbf{Scalar skew dataset} :
We tested on two sets of scalar skew data. As in \cite{DeWitt92vldb}, we fix an output size and vary the values of $M$ and $N$ to examine the effect of different key skewness in the two given tables. For the first dataset, we set $M = 10^5$, and $N=2 \times 10^4$.
For the second set, we set $M = 2*10^5$ and $N = 10^4$.
The output size of the join of $S$ and $T$ for both datasets
is 190GB. In both datasets, $|S|=|T|=1.5M$, and the skew factor $\sigma$ is 600.

\subsubsection{Runtime Analysis}

The total runtimes are shown in Figure \ref{exp:JoinZipftime} and \ref{exp:Joinskewtime}. It can be seen that we achieve almost linear speedup, going from 3 to 7 processes almost halved the runtime and similarly going from 7 to 15 processes.
This is a result of the highly even workload distribution. 

 Note that the speedup effect beyond 15 processes is discounted by the overhead in the file replication of Hadoop HDFS, since the default file replication of 3 is adopted. A single I/O in a sequential algorithm becomes 3 I/O's for the parallel algorithm on three hard disks. Hence with only 34 hard disks in total, there is good speedup effects with only up to 15 processes.

\subsubsection{Workload Imbalance}

The results of workload distribution are shown in Figures \ref{exp:JoinZipfwork} and \ref{exp:Joinskewload}.
For the scalar skew dataset, RandJoin did not distribute the workload as evenly when the number of processor, $t$, is large. The reason is that in such cases, the value of $M/a$ and $N/b$ are not big enough to satisfy the condition in Corollary \ref{cor1}. According to the algorithm, the values of $a$ and $b$ are set as follows:
\begin{center}
\begin{small}
\begin{tabular}{|l|c|c|c|c|c|c|c|}
  \hline
  number of  processes & 3 & 7 & 15 & 30 & 60 & 120 & 180 \\ \hline
$a$ & 1 & 1 & 3 & 5 & 6 & 12 & 12 \\ \hline
$b$ & 3 & 7 & 5 & 6 & 10 & 10 & 15 \\
  \hline
\end{tabular}
\end{small}
\end{center}
When $b$ is above 10, and $N = 10^4$, $N/b$ will be less than 300 and the condition in
Corollary \ref{cor1} will be violated.

For StatJoin, the workload distribution is not as even for the Zipf distributed data as for the scalar skew data.
The uneven distribution occurs for $\theta = 1$ (uniformly distribution). This is because the domain of the key is [1000, 1999].
Hence, the join result size for each key is around $|S|/1000 \times |T|/1000$. For $|S|=|T|=5M$, the estimated size is 2.5M, which is quite large. Since no skew key exists, all join results will be small join results, but they are not small and can introduce up to 2.5M imbalance to the workload, and this is what we observed from the experiments.
For $\theta = 0$, when $t$ is small, $W/t$ can be large, and therefore the join for some skew keys may also become small join results. However, these are exceptional cases and StatJoin performed very well for the cases when we have skew keys and more processors.

\subsubsection{Statistics Collection for StatJoin}

As described in Section \ref{sec:join}, when compared to RandJoin, StatJoin requires two more steps (Steps 1 and 2) for sorting and statistics collection.
Hence, in addition to the total running time for RandJoin and StatJoin, we further analyze the running time and percentage of total runtime used for statistics collection in StatJoin.
Table \ref{tab:compZipf} and \ref{tab:compScalar} show the results for Zipf skew dataset where $\Theta=0$ and scalar skew dataset where $M=2 \times 10^5$, respectively.
Figure \ref{exp:TimeCom} shows the running time comparison for Tables \ref{tab:compZipf} and \ref{tab:compScalar}. Note that the running time in Figure \ref{exp:TimeCom} is in $\log$ scale.

\begin{table}
\begin{center}
\begin{small}
\begin{tabular}{|l|l||c|c|c|c|}
  \hline
  \multicolumn{2}{|c|}{number of  processes} & 3 & 7 & 15 & 30  \\ \hline \hline
RandJoin & $Total$ & 5119 & 2777 & 1321 & 985  \\ \hline
\multirow{3}{*}{StatJoin}  & $Total$ & 5010 & 2516 & 1370 & 981 \\
\cline{2-6}
\cline{2-6}
& $Statistics$ & 31 & 52 & 39 & 36 \\
\cline{3-6}
& $Collection$ & 0.6\% & 2.1\% & 2.8\% & 3.7\% \\
  \hline
\end{tabular}
\end{small}
\caption{Running time comparison(in sec) for Zipf skew dataset($\Theta=0$)} \label{tab:compZipf}.
The bottom row shows the percentage of the total runtime taken for statistics collection.
\end{center}
\end{table}

\begin{table}
\begin{center}
\begin{small}
\begin{tabular}{|l|l||c|c|c|c|}
  \hline
  \multicolumn{2}{|c|}{number of  processes} & 3 & 7 & 15 & 30  \\ \hline \hline
RandJoin & $Total$ & 1532 & 909 & 565 & 532  \\ \hline
\multirow{3}{*}{StatJoin}  & $Total$ & 1557 & 883 & 572 & 587 \\
\cline{2-6}
& $Statistics$ & 35 & 51 & 39 & 31 \\
\cline{3-6}
& $Collection$ & 2.3\% & 5.7\% & 6.9\% &5.3\% \\
  \hline
\end{tabular}
\end{small}
\caption{Running time comparison(in sec) for scalar skew dataset($M=2 \times 20^5$)} \label{tab:compScalar}.
The bottom row shows the percentage of the total runtime taken for statistics collection.
\end{center}
\end{table}

\begin{figure}[!t]
\begin{center}
\hspace*{-2mm}
\includegraphics[width = 1.65in,height=1.6in]{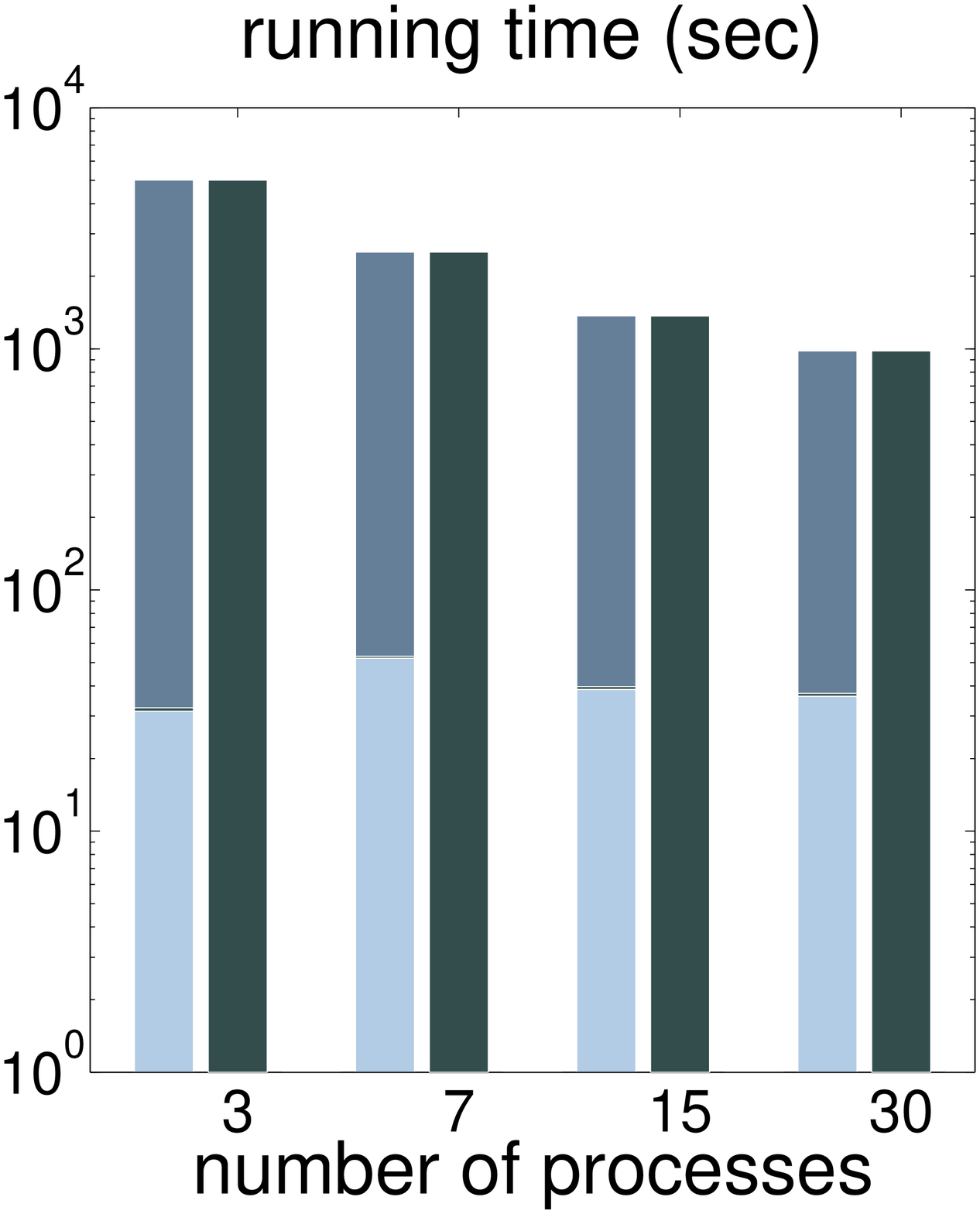}
\includegraphics[width = 1.65in,height=1.6in]{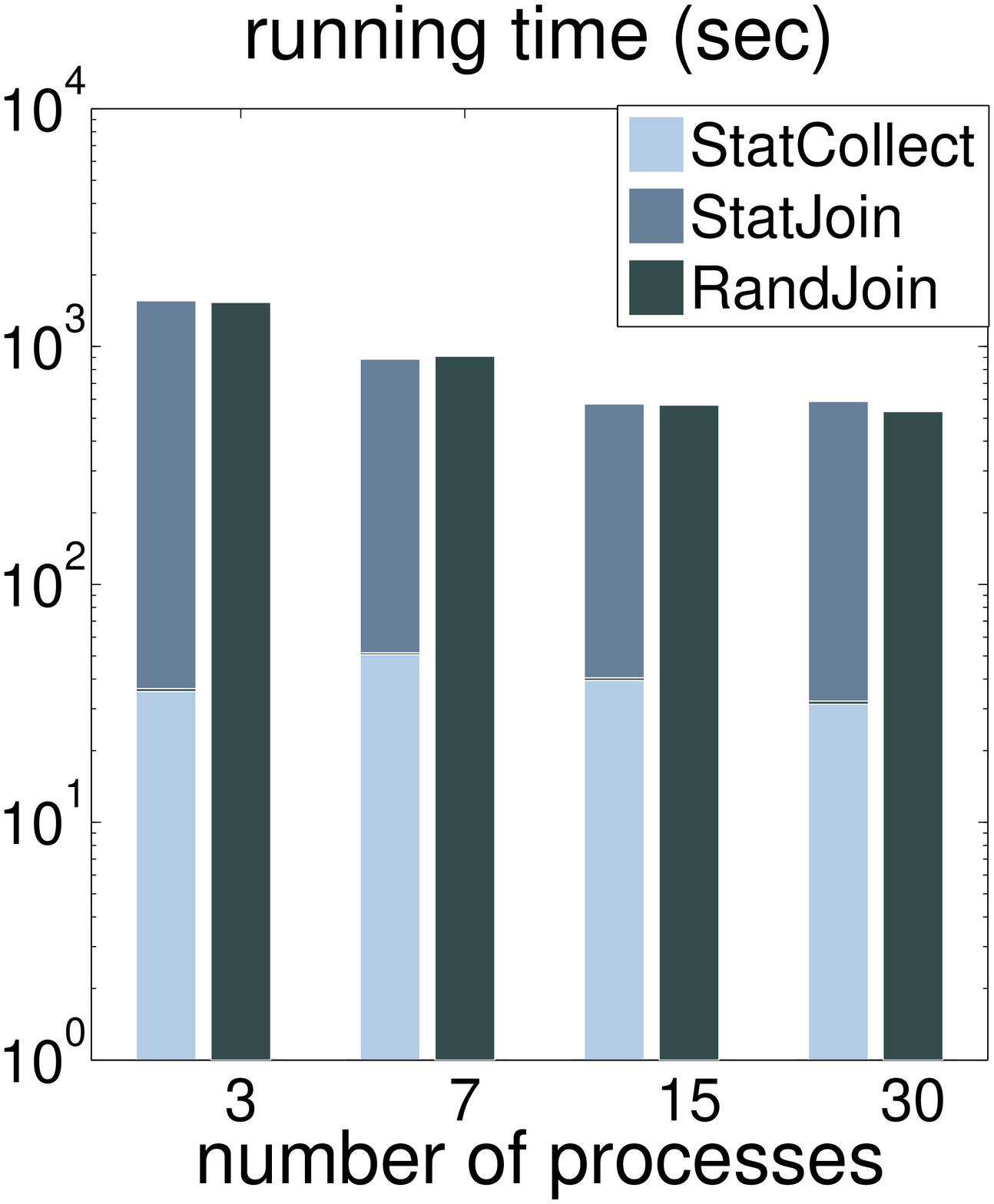}
\\
\hspace*{10mm} (a) Zipf
\hspace*{16mm}
(b) Scalar Skew $(M = 2\times10^5)$
\hspace*{15mm}\\
\caption{Time Comparison (in sec): First bar is for StatJoin, second bar is for RandJoin in each group} \label{exp:TimeCom}
\end{center}
\end{figure}

As indicated in Tables \ref{tab:compZipf} and \ref{tab:compScalar}, the running time for Steps 1 and 2 in StatJoin is a very small percentage of the total running time.
This is mainly because the output size dominates the problem size, and Steps 1 and 2 deal only with the input.
Also we notice that the total running times for RandJoin and StatJoin are similar. This is due to the fact that the shuffling process is sensitive to the sorted ordering of the keys. The shuffling of Round 3 in StatJoin is faster than that of RandJoin since keys have been sorted after Rounds 1 and 2 in StatJoin. Thus, the statistics collection steps are useful for later processing.
We conclude that both RandJoin and StatJoin are highly effective in the parallel computation of skew join, with little overhead in synchronization and close to optimal workload distribution.

%

%
\section{Related Work}\label{sec:related}

Terasort \cite{Malley08Yahoo} has won the Jim Gray's benchmark sorting competition in 2009. The idea is to randomly sample the given data objects and determine the distribution of data objects to machines based on the sampled objects. However, in the original algorithm, it is not clear how the number of samples should be determined. This problem is studied in \cite{Tao13sigmod}. Given $n$ objects and $t$ machines, it is found that when the number of samples at each machine is set to $\ln(nt)$, the workload is $O(n/t)$ at each machine. In their empirical studies, this refined version of Terasort is compared with the default sorting algorithm in Hadoop.

With Hadoop's algorithm,
given $k$ blocks of input data, the master node gathers the first $\lceil 10^5/k \rceil$ data objects of each block and form a sample set. The sample set is sorted. The master determines boundary points as in Terasort by picking the point $b_i$ to be the
$i \lceil 10^5/t \rceil$-th smallest object in the sample set.
Hence, it differs from Terasort in that the sample set is not selected randomly, but selected from the beginning of each data block. Consequently, the result will be highly dependent on the data distribution in the input. When the data set is highly skew, Hadoop's algorithm will introduce highly unbalanced workloads. Although Terasort can avoid the problem of skew input, the random sampling does not give close to optimal load balancing. From experiments in \cite{Tao13sigmod}, in most cases, the maximum workload of a machine exceeds the optimal result by over 50\%. Our proposed algorithm SMMS replaces the randomized sampling process with a deterministic bucket boundary computation, which gives rise to better
theoretical and experimental results.

A theta-join algorithm is proposed in \cite{Okcan11sigmod} which also assigns tuples to machines in a random manner.
The authors model the join of $S$ and $T$ with a join-matrix.
However, their algorithm requires the computation of a matrix-to-reducer mapping which assigns regions of the join matrix to machines, and the assignment cannot be made even in general.
With RandJoin, the assignment of tuples to machines is based on a tuple-to-interval mapping which is simpler and is guaranteed an expected even mapping.

A number of projects have implemented SQL-like language translators,
 integrating database query constructs
on MapReduce to support database operations and query optimization. Examples of
such projects include the SCOPE project at
Microsoft \cite{Chaiken08vldb} , YSmart \cite{Lee11icdcs}, Tenzing \cite{Chattopadhyay11vldb},
 open source HIVE \cite{Thusoo10icde}, and the Pig at Yahoo\cite{Olston08sigmod}.
 To the best of our knowledge, only Apache Pig supports skew join. However, the Pig solution does not provide a guarantee of load balancing, and workload will not be balanced if both tables contain skew keys, as pointed out in \cite{Gates13Bieee}. 

Multiway join in Map-Reduce has been studied in \cite{Afrati10edbt}, with a focus
on query optimization by means of query plan selection with respect to minimal input replication cost.
The authors of \cite{Vernica10sigmod} 
considered a special type of similarity join in MapReduce and proposed techniques for limited memory.
Theta join has been considered in \cite{Okcan11sigmod,Zhang12pvldb}.
Efficient processing of k-nearest neighbor joins using MapReduce is considered in
\cite{Lu12pvldb}.

As noted in \cite{Tao13sigmod}, while MapReduce or parallel algorithms in general for computer clusters have aimed at load balancing, minimization of space, CPU, I/O and network costs, there have been no systematic constraints on the requirements 
or analysis for such algorithm design.
In \cite{Leiserson10SPAA}, {\emph{work efficiency}} is considered for MapReduce.
Work efficiency has been defined for parallel algorithms, which are said to be work efficient if the total number of operations in the parallel execution is the same to within a constant factor as that of a comparable serial algorithm.
The constant factor is called the work efficiency.
However, work efficiency does not correspond directly to the overall runtime efficiency when
the execution is not evenly distributed, or when there exist dependencies among the jobs assigned to different machines, or when communication cost is substantial.
There also exist other works such as \cite{Okcan11sigmod} which focus only on different aspects of load balancing.

To address this lack of a comprehensive yardstick, the notion of a minimal MapReduce algorithm is introduced in \cite{Tao13sigmod}.
Let $S$ be the set of input objects, $n$ be the number of objects in $S$,
and $t$ be the number of machines. Define $m = n/t$, hence $m$ is the number of objects per machine when $S$ is evenly distributed. A minimal MapReduce algorithm by definition satisfies four
criterion: (1) $O(m)$ storage, (2) each machine sends and receives $O(m)$
words, (3) constant number of rounds, and (4) every machines takes $O(T_{seq}/t)$
computation time, where $T_{seq}$ is the time needed to solve the same problem on a single machine by a comparable sequential algorithm.

Our model of an $(\alpha,k)$-minimal algorithm also considers the workload distribution as an important factor. We make explicit the number of rounds in the algorithm and also a bound on the maximum workload and network transmission costs. Such quantifiers help to give a clearer indicator for the guarantee of the algorithm.
Our model can be readily applied to certain known parallel algorithms
such as the sorting algorithm PSRS
(Parallel Sorting by Regular Sampling) \cite{Shi92jpdc,Li93jpdc}.
In \cite{Okcan11sigmod} a MapReduce algorithm is proposed for the computation of cross-product of two tables $S$ and $T$, and it is shown that no reducer produces more than $4|S||T|/t$ tuples for $t$ reducers. This algorithm can be shown to be
$(2,(4+1/\sigma))$-minimal if $\sigma$ is the skew factor.
Similarly, we believe that other algorithms can be shown to be $(\alpha, k)$-minimal for particular values of $\alpha$ and $k$.
However, this model does not apply to parallel algorithms with no explicit rounds in the computation. Hence, other models may be needed for the analysis of such algorithms.

\section{Conclusion}\label{sec:conclusion}

We introduce the concept of $(\alpha,k)$-minimality for the analysis of MapReduce or MPI algorithms. An $(\alpha,k)$-minimal algorithm consists of $\alpha$ rounds and if
even workload distribution is $m$, each machines has at most $km$ workload.  We study the fundamental problems of sorting and skew join. We derived algorithms for both problems that achieve the best known theoretical guarantees on even workload distribution.
Our proposed algorithms are $(\alpha,k)$-minimal for $\alpha \leq 3$ and $k \leq 2$. Extensive empirical study shows that our sorting algorithm performs better than the state-of-the-art method of TeraSort. All our algorithms achieve near optimal workload distribution in all test cases and the results substantiate our theoretical analysis.

\vspace*{3mm}

\bibliographystyle{abbrv}
\bibliography{refMapReduce}

\end{document}